\RequirePackage{fix-cm}
\documentclass{svjour3}                     
\smartqed  
\usepackage{ae,aecompl}

\usepackage[T1]{fontenc}
\usepackage[latin9]{inputenc}
\usepackage{color}
\usepackage[british]{babel}
\usepackage{float}
\usepackage{textcomp}
\usepackage{mathtools}
\usepackage{algorithm2e}
\usepackage{amsmath}
\usepackage{amssymb}
\usepackage{mathdots}
\usepackage{graphicx}
\usepackage{xcolor}
\usepackage{esint}
\usepackage{multirow}
\usepackage[unicode=true,pdfusetitle,
 bookmarks=true,bookmarksnumbered=false,bookmarksopen=false,
 breaklinks=false,pdfborder={0 0 1},backref=false,colorlinks=true,linkcolor=blue,citecolor=darkgreen,urlcolor=blue]
 {hyperref}
\usepackage{algpseudocode}
\usepackage{url}
\usepackage{subfigure}
\usepackage{lscape}
\usepackage{enumitem}
\usepackage[margin=3.5cm]{geometry}
\smartqed

\newcommand{\qede}{\hspace*{\fill}$\Diamond$\medskip}

\makeatletter
\newcommand{\mylabel}[2]{#2\def\@currentlabel{#2}\label{#1}}
\makeatother

\makeatletter

\definecolor{darkgreen}{RGB}{0,150,0}

\newtheorem{fact}{\protect\factname}

\algnewcommand\And{\textbf{and}}
\providecommand{\factname}{Fact}
\makeatother

\begin{document}

\title{Local convergence of the Levenberg--Marquardt method under H\"{o}lder metric subregularity\thanks{F.J. Arag\'on was supported by MINECO of Spain
 and ERDF of EU, as part of the Ram\'on y Cajal program (RYC-2013-13327) and the I+D grant MTM2014-59179-C2-1-P. M. Ahookhosh, R.M.T. Fleming, and
 P.T. Vuong were supported by the U.S. Department of Energy, Offices of Advanced Scientific Computing Research and the Biological and Environmental
 Research as part of the Scientific Discovery Through Advanced Computing program, grant \#DE-SC0010429. P.T. Vuong was also supported by the Austrian
 Science Foundation (FWF), grant I 2419-N32.}}
\authorrunning{M. Ahookhosh, F.J. Arag\'on Artacho, R.M.T. Fleming, P.T. Vuong}

\institute{M. Ahookhosh \at Systems Biochemistry Group, Luxembourg Center for Systems Biomedicine,
University of Luxembourg, Campus Belval, 4362 Esch-sur-Alzette, Luxembourg.\\
Department of Electrical Engineering (ESAT-STADIUS) - KU Leuven, Kasteelpark Arenberg 10,
3001 Leuven, Belgium.\\
\email{masoud.ahookhosh@kuleuven.be}
\and
F.J. Arag\'on Artacho \at Department of Mathematics, University of Alicante, Spain. \email{francisco.aragon@ua.es}
\and
R.M.T. Fleming \at Systems Biochemistry Group, Luxembourg Center for Systems Biomedicine,
University of Luxembourg, Campus Belval, 4362 Esch-sur-Alzette, Luxembourg. \email{ronan.mt.fleming@gmail.com}
\and
P.T. Vuong \at Systems Biochemistry Group, Luxembourg Center for Systems Biomedicine,
University of Luxembourg, Campus Belval, 4362 Esch-sur-Alzette, Luxembourg.\\
Faculty of Mathematics, University of Vienna, Oskar-Morgenstern-Platz 1, 1090 Vienna, Austria.\\
\email{vuong.phan@univie.ac.at}
}

\author{Masoud Ahookhosh \and Francisco J. Arag\'on Artacho \and Ronan M.T. Fleming \and Phan T. Vuong}

\date{}

\maketitle

\begin{abstract}
We describe and analyse Levenberg--Marquardt methods for solving systems of nonlinear equations. More specifically, we propose an adaptive formula for the Levenberg--Marquardt
parameter and analyse the local convergence of the method under
H\"{o}lder metric subregularity of the function defining the equation and H\"older continuity of its gradient mapping. Further, we analyse the local convergence
of the method under the additional assumption that the \L{}ojasiewicz gradient inequality holds. We finally report
encouraging numerical results confirming the theoretical findings for the
problem of computing moiety conserved steady states in biochemical reaction
networks. This problem can be cast as finding a solution of a system of nonlinear
equations, where the associated mapping satisfies the \L{}ojasiewicz gradient inequality assumption.
\end{abstract}

\section{Introduction}\label{sec.1}

For a given continuously differentiable mapping $h:\mathbb{R}^{m}\rightarrow\mathbb{R}^{n}$,
we consider the problem of finding a solution of the system of nonlinear
equations
\begin{equation}
h(x)=0,\quad x\in\mathbb{R}^{m}.\label{eq:nonequa}
\end{equation}
We denote by $\Omega$ the set of solutions of this problem, which
is assumed to be nonempty. Systems of nonlinear equations of type~\eqref{eq:nonequa}
 frequently appear in the mathematical modelling of many real-world
applications in the fields of solid-state physics~\cite{Eilenberger},
quantum field theory, optics, plasma physics~\cite{hasegawa_plasma_1975},
fluid mechanics~\cite{Whitham}, chemical kinetics~\cite{artacho_globally_2014,artacho_accelerating_2015},
and applied mathematics including the discretisation of ordinary and
partial differential equations~\cite{ortega_iterative_2000}.

A classical approach for finding a solution of~\eqref{eq:nonequa}
is to search for a minimiser of the nonlinear least-squares problem
\begin{equation}
\min_{x\in\mathbb{R}^{m}}\psi(x),\quad\text{with }\psi:\mathbb{R}^{m}\rightarrow\mathbb{\mathbb{R}}\text{ given by }\psi(x):=\frac{1}{2}\|h(x)\|^{2},\label{eq:meritfunc}
\end{equation}
where $\|\cdot\|$ denotes the Euclidean norm. This is a well-studied
topic and there are many iterative schemes with fast local convergence
rates (e.g., superlinear or quadratic) such as Newton, quasi-Newton,
Gauss--Newton, adaptive regularised methods, and the Levenberg--Marquardt
method. When $m=n$, to guarantee fast local convergence, these methods require
an initial point $x_{0}$ to be sufficiently close to a solution $x^{*}$,
and the \emph{matrix gradient} of $h$ at $x^{*}$ (i.e., the transpose of the Jacobian matrix), denoted by $\nabla h(x^{*})$, to be
\emph{nonsingular} (i.e., full rank), cf.~\cite{BelCGMT,fischer2015globally,nocedal_numerical_2006-1,ortega_iterative_2000,Yuan2015}.

The Levenberg--Marquardt method is a standard technique
used to solve the nonlinear system~\eqref{eq:nonequa}, which is a combination of the
gradient descent and the Gauss--Newton methods. More precisely,
in each step, for a positive parameter $\mu_{k}$, the convex subproblem
$$\min_{d\in\mathbb{R}^{m}}\phi_{k}(d),$$
with $\phi_{k}:\mathbb{R}^{m}\rightarrow\mathbb{\mathbb{R}}$ given by
\begin{equation}
\phi_{k}(d):=\left\|\nabla h(x_{k})^{T}d+h(x_{k})\right\|^{2}+\mu_{k}\|d\|^{2},\label{eq:def_phi_k}
\end{equation}
is solved to compute a direction $d_{k}$, which is the
unique solution to the system of linear equations
\begin{equation}
\left(\nabla h(x_{k})\nabla h(x_{k})^{T}+\mu_{k}I\right)d_{k}=-\nabla h(x_{k})h(x_{k}),\label{eq:LM_direction}
\end{equation}
where $I\in\mathbb{\mathbb{R}}^{m\times m}$ denotes the identity
matrix. By choosing a suitable parameter $\mu_{k}$, the Levenberg--Marquardt
method acts like the gradient descent method whenever the current
iteration is far from a solution $x^{*}$, and behaves similar to the Gauss--Newton
method if the current iteration is close to $x^{*}$. The parameter
$\mu_{k}$ helps to overcome problematic cases where $\nabla h(x_{k})\nabla h(x_{k})^{T}$
is singular, or nearly singular, and thus ensures the existence of a unique solution to~\eqref{eq:LM_direction}, or avoids very large steps, respectively. For $m=n$, the Levenberg--Marquardt method is known to be quadratically
convergent to a solution of~\eqref{eq:nonequa} if $\nabla h(x^{*})$ is nonsingular. In fact,
the nonsingularity assumption implies that the solution to the minimisation problem~\eqref{eq:meritfunc}
must be locally unique, see~\cite{bellavia_strong_2015,kanzow_levenbergmarquardt_2004,alefeld_rate_2001}.
However, assuming local uniqueness of the solution might be restrictive
for many applications.

The notion of (\emph{local}) \emph{error bound} usually plays a key
role in establishing the rate of convergence of the sequence of iterations
generated by a given algorithm. This condition guarantees that the
distance from the current iteration $x_{k}$ to the solution set $\Omega$, denoted by
$\mathrm{dist}(x_{k},\Omega)=\inf_{y\in \Omega}\|x_k-y\|$, is less than the value of a residual
function $R:\mathbb{R}^m\to\mathbb{R_+}$ at that point ($R(x_{k})$). 
The
earliest publication using error bounds for solving a linear inequality
system is due to Hoffman~\cite{Hoffman_1952}, which was followed
by many other authors, especially in optimisation. For more information
about error bounds, we recommend the nice survey~\cite{Pang_1997}.

For the particular case of nonlinear systems of equations, Yamashita
and Fukushima~\cite{alefeld_rate_2001} proved the local quadratic
convergence of the Levenberg--Mar\-quardt method with $\mu_{k}=\|h(x_{k})\|^{2}$
assuming a local error bound condition. More precisely, they assumed \emph{metric subregularity} of $h$ around $(x^{*},0)$, which entails
the existence of some constants $\beta>0$  and $r>0$ such that
\begin{equation}
\beta\,\mathrm{dist}(x,\Omega)\leq\|h(x)\|,\quad\forall x\in\mathbb{B}(x^{*},r),\label{eq:leb}
\end{equation}
where $\mathbb{B}(x^{*},r)$ denotes the closed ball centered at $x^{*}$
with radius $r>0$. In this case, the residual function is given by
$R(x):=\frac{1}{\beta}\|h(x)\|$. In those situations where the value of $\beta$ is known,
the condition $\|h(x)\|\leq\varepsilon$ can be
used as a stopping criterion for an iterative scheme, as it entails
that the iterations must be close to a solution of~\eqref{eq:nonequa}.

Let us emphasise that, for $m=n$, the nonsingularity of $\nabla h(x^{*})$ implies that
$x^{*}$ is locally unique and that~\eqref{eq:leb} holds.
Indeed, by the Lyusternik--Graves theorem (see, e.g.,~\cite[Theorem~5D.5]{DontchevRockafellar2014}, \cite[Theorem~1.57]{Mordukhovich2006}, or~\cite[Proposition~1.2]{CDK18}), the nonsingularity of
$\nabla h(x^{*})$ is equivalent to the strong metric regularity of
$h$ at $(x^*,0)$, which implies strong metric subregularity of
$h$ at $(x^*,0)$. However, the latter does not imply the nonsingularity assumption
and allows the solutions to be locally nonunique. This means that metric subregularity is a weaker assumption than the nonsingularity. In fact, for $m$ possibly different than $n$, strong metric subregularity of $h$ at $(x^*,0)$ is equivalent to surjectivity of $\nabla h(x^*)$ (see, e.g., \cite[Proposition~1.2 and Theorem~2.6]{CDK18}). The successful use
of the local error bound has motivated many researchers
to investigate, under assumption~\eqref{eq:leb}, the local convergence of trust-region
methods~\cite{Fan2006}, adaptive regularised methods~\cite{bellavia_strong_2015},
and the Levenberg--Marquardt method~\cite{behling_effect_2013,fan_modified_2012,fan_quadratic_2005},
among other iterative schemes.

The main motivation for this paper comes from a nonlinear system of
equations, the solution of which corresponds to a steady state of a given biochemical reaction
network, which plays a crucial role in the modeling of biochemical
reaction systems. These problems are usually ill-conditioned and require the application of the Levenberg--Marquardt method. As we numerically show in Section~\ref{sec.numapp}, $\nabla h$ is usually rank deficient at the solutions of~\eqref{eq:nonequa}. During our study of the properties of this problem,
we were not able to show that the metric subregularity condition~\eqref{eq:leb}
is satisfied. However, taking standard biochemical assumptions \cite{artacho_accelerating_2015}, we can  show that the corresponding merit function is real analytic and thus satisfies the \L{}ojasiewicz gradient inequality and is H\"{o}lder metrically subregular around the solutions.

The local convergence of a Levenberg--Marquardt method under H\"older metric subregularity has been recently studied in~\cite{Guo2015,ZhuLin16}. Nonetheless, the standard rules for the regularisation parameter have a very poor performance when they are applied for solving the nonlinear equation arising from the biochemical reaction network systems, as we show in a numerical experiment in Section~\ref{sec.numapp}. This motivated our quest to further investigate an adaptive Levenberg--Marquart method under the assumption that the underlying mapping is H\"older metrically subregular.

From the definition of the Levenberg--Marquardt direction
in~\eqref{eq:LM_direction}, we observe that a key factor in the
performance of the Levenberg--Marquardt method is the choice
of the parameter $\mu_{k}$, cf.~\cite{izmailov2014newton,Kelley1}. Several parameters
have been proposed to improve the efficiency of the method. For example,
Yamashita and Fukushima~\cite{alefeld_rate_2001} took $\mu_{k}=\|h(x_{k})\|^{2}$,
Fischer~\cite{fischer2002local} used $\mu_{k}=\|\nabla h(x_{k})h(x_{k})\|$,
while Fan and Yuan~\cite{fan_quadratic_2005} proposed $\mu_{k}=\|h(x_{k})\|^{\eta}$
with $\eta\in[1,2]$. Ma and Jiang~\cite{MJ07} proposed a convex combination of these two types of parameters, namely, $\mu_k= \theta\|h(x_{k})\|+(1-\theta)\|\nabla h(x_{k})h(x_{k})\|$ for some constant $\theta\in{[0,1]}$. In a subsequent work, Fan and Pan~\cite{fan_pan} proposed the more general choice $\mu_k=\xi_k\rho(x_k)$, where $\xi_k$ is updated by a trust-region technique, $\rho(x_k)=\min\left\{\widetilde{\rho}(x_k),1\right\}$ and $\widetilde{\rho}:\mathbb{R}^m\to\mathbb{R}_+$ is a positive function such that $\widetilde{\rho}(x_k)=O\left(\|h(x_k)\|^\eta\right)$, with $\eta\in{]0,2]}$. Inspired by these works, and assuming that the function $h$ is H\"{o}lder
metrically subregular of order $\delta\in\ ]0,1]$ and its gradient $\nabla h$ is H\"older continuous of order $\upsilon\in{]0,1]}$, in this paper we consider an adaptive
parameter of the form
\begin{equation}
\mu_{k}:=\xi_{k}\|h(x_{k})\|^{\eta}+\omega_{k}\|\nabla h(x_{k})h(x_{k})\|^{\eta},\label{eq:muk}
\end{equation}
where $\eta>0$, $\xi_{k}\in[\xi_{\min},\xi_{\max}]$ and $\omega_k\in[\omega_{\min},\omega_{\max}]$, for some  constants  $0\leq\xi_{\min}\leq\xi_{\max}$ and $0\leq\omega_{\min}\leq\omega_{\max}$ such that $\xi_{\min}+\omega_{\min}>0$.

In our first main result, Theorem~\ref{thm:dist_con_sup}, we provide an interval depending on $\delta$ and $\upsilon$ where the parameter~$\eta$ must be chosen to guarantee the  superlinear convergence of the sequence generated by the Levenberg--Marquardt method with the adaptive parameter~\eqref{eq:muk}. In our second main result, Theorem~\ref{thm:convergence_of_dist(x,Omega)_general}, under the additional assumption that the merit function $\psi$ defined in~\eqref{eq:meritfunc} satisfies the \L{}ojasiewicz gradient inequality with exponent $\theta\in{]0,1[}$, we prove local convergence for every parameter $\eta$ smaller than a constant depending on both $\upsilon$ and $\theta$. As a consequence, we can ensure local convergence of the Lebenverg--Marquardt algorithm to a solution of~\eqref{eq:nonequa} for all the above-mentioned biochemical networks as long as the parameter $\eta$ is chosen sufficiently small. To the best of our knowledge, this is the first such algorithm able to reliably handle these nonlinear systems arising in the study of biological networks. We successfully apply the proposed algorithm to nonlinear systems derived from
many real biological networks, which are representative of a diverse set of biological species.

The remainder of this paper is organised as follows. In the next section,
we particularise the H\"{o}lder metric subregularity for nonlinear equations and recall the \L{}ojasiewicz inequalities.
We investigate the local convergence of the Levenberg--Marquardt method under these conditions in Section~\ref{sec:Local}. In Section~\ref{sec.numapp}, we report encouraging numerical
results where nonlinear systems, arising from biochemical reaction networks, were quickly solved.
Finally, we deliver some conclusions in Section~\ref{sec:conclusion}.

\section{H\"{o}lder metric subregularity and \L ojasiewicz inequalities}\label{sec.2}

 Let us begin this section by recalling the notion of
H\"{o}lder metric subregularity, which can be also defined in a similar
manner for set-valued mappings (see, e.g.,~\cite{kruger_error_2015,CDK18}).
\begin{definition}
\label{def:hmsr}A mapping $h:\mathbb{R}^{m}\rightarrow\mathbb{R}^{n}$
is said to be \emph{H\"{o}lder metrically subregular} of order $\delta>0$
around $(\overline{x},\overline{y})$ with $\overline{y}=h(\overline{x})$ if there exist
some constants $r>0$ and $\beta>0$ such that
\[
\beta\,\mathrm{dist}\!\left(x,h^{-1}(\overline{y})\right)\leq\|\overline{y}-h(x)\|^{\delta},\quad\forall x\in\mathbb{B}(\overline{x},r).
\]
\end{definition}
For any solution $x^{*}\in\Omega$ of the system of nonlinear equations~\eqref{eq:nonequa},
the H\"{o}lder metric subregularity of $h$ around $(x^{*},0)$ reduces to
\begin{equation}
\beta\,\mathrm{dist}(x,\Omega)\leq\|h(x)\|^{\delta},\quad\forall x\in\mathbb{B}(x^{*},r).\label{eq:errorBound}
\end{equation}
Therefore, this property provides an upper bound for the distance from
any point sufficiently close to the solution $x^{*}$ to the nearest
zero of the function.

H\"{o}lder metric subregularity around $(x^{*},0)$ is also called \emph{H\"{o}lderian
local error bound}~\cite{ngai_global_2015,vui_global_2013}. It is
known that H\"{o}lder metric subregularity is closely related to the \L ojasiewicz
inequalities, which are defined as follows.
\begin{definition}
Let $\psi:U\to\mathbb{R}$ be a  function defined on an
open set $U\subseteq\mathbb{R}^{m}$, and assume that the set of zeros
$\Omega:=\left\{ x\in U,\,\psi(x)=0\right\} $ is nonempty.
\begin{itemize}
\item [(i)] The function $\psi$ is said to satisfy the \emph{\L{}ojasiewicz inequality} if for every compact subset $C\subset U$,
there exist positive constants $\varrho$ and $\gamma$ such that
\begin{equation}
\mathrm{dist}(x,\Omega)^{\gamma}\leq \varrho|\psi(x)|,\quad\forall x\in C.\label{eq:LojaIne}
\end{equation}
\item [(ii)] The function $\psi$ is said to satisfy the \emph{\L{}ojasiewicz gradient inequality} if for any critical point~$x^*$, there exist constants  $\kappa>0,\varepsilon>0$
and $\theta \in{]0,1[}$  such that
\begin{equation}
|\psi(x)-\psi(x^*)|^{\theta}\leq \kappa \|\nabla \psi(x)\|,\quad\forall x\in\mathbb{B}(x^*,\varepsilon).\label{eq:Lojasiewicz_Gradient_Inequality}
\end{equation}
\end{itemize}
 \end{definition}
Stanis\l{}aw \L{}ojasiewicz proved that every real analytic function satisfies these properties~\cite{lojasiewicz1965ensembles}. Recall
that a function $\psi:\mathbb{R}^{m}\to\mathbb{R}$ is said to be
\emph{real analytic} if it can be represented by a convergent power
series. Fortunately, real analytic functions frequently appear in
real world application problems. A relevant example in biochemistry
is presented in Section~\ref{sec.numapp}.
\begin{fact}
[{\cite[pp. 62 and 67]{lojasiewicz1965ensembles}}]\label{prop:Lojasiewicz Inequality-1}Every
real analytic function $\psi:\mathbb{R}^{m}\to\mathbb{R}$  satisfies both
the \L{}ojasiewicz inequality and the \L{}ojasiewicz gradient inequality.
\end{fact}
Clearly, if the merit function $\psi(\cdot)=\frac{1}{2}\|h(\cdot)\|^{2}$
satisfies the \L ojasiewicz inequality~\eqref{eq:LojaIne}, then
the mapping $h$ satisfies~\eqref{eq:errorBound} with $\beta:=(2/\varrho)^{1/\gamma}$
and $\delta:=2/\gamma$; i.e., $h$ is H\"{o}lder metrically subregular
around $(x^{*},0)$ of order $2/\gamma$. In addition, if $\psi(\cdot)$ satisfies the \L{}ojasiewicz gradient inequality~\eqref{eq:Lojasiewicz_Gradient_Inequality}, then for any $\overline{x} \in \Omega $ and $x\in\mathbb{B}(\overline{x},\varepsilon)$, it holds
\begin{eqnarray*}
\frac{1}{\varrho}\mathrm{dist}(x,\Omega)^\gamma
\leq  |\psi(x)|
\leq \kappa^{1/\theta}\|\nabla \psi(x)\|^{1/\theta}
= \kappa^{1/\theta}\|\nabla h(x)h(x)\|^{1/\theta}.
\end{eqnarray*}


The \L{}ojasiewicz gradient inequality has recently gained much attention because of its role for proving the convergence of various numerical methods (e.g., \cite{BDL07,AB09,ABS13,artacho_accelerating_2015}). The connection between this property and metric regularity of the set-valued mapping $\Psi(x):=[\psi(x),\infty[$ on an adequate set was revealed in~\cite{BDLM10}, where it was also applied to deduce strong convergence of the proximal algorithm.

In some cases, for example when $\psi$ is a polynomial with an isolated
zero at the origin, an order of the H\"{o}lder metric subregularity is
known~\cite{gwozdziewicz_lojasiewicz_1999,kurdyka_separation_2014,li_holder_2012}.
\begin{fact}
[{\cite[Theorem~1.5]{gwozdziewicz_lojasiewicz_1999}}]\label{thm:Loj_exponent}Let
$\psi:\mathbb{R}^m\to\mathbb{R}$ be a polynomial function with an
isolated zero at the origin. Then $\psi$ is H\"{o}lder metrically subregular
around $(0,0)$ of order $\left((\mathrm{deg}\, \psi-1)^m+1\right)^{-1}$, where $\mathrm{deg}\,\psi$
denotes the degree of the polynomial function $\psi$.
\end{fact}

The next example shows that the Powell singular function, which is a classical test function for nonlinear systems of equations, is not metrically subregular around its unique solution but is H\"older metrically subregular there.  In addition, it demonstrates that the order given by Fact~\ref{thm:Loj_exponent} is, in general, far from being tight.
\begin{example}\label{ex:Powell}
The Powell singular function~\cite{more_1981}, which is
the function $h:\mathbb{R}^{4}\to\mathbb{R}^{4}$ given by
\begin{equation*}
h(x_{1},x_{2},x_{3},x_{4}):=\left(x_{1}+10x_{2},\sqrt{5}(x_{3}-x_{4}),(x_{2}-2x_{3})^{2},\sqrt{10}(x_{1}-x_{4})^{2}\right),\label{eq:powell}
\end{equation*}
is (strongly) H\"{o}lder metrically subregular around $(0_4,0)$ but does not satisfy
the metric subregularity condition~\eqref{eq:leb}. We have $\Omega=\left\{ 0_4\right\} $
and $\nabla h(0_4)$ is singular; thus, $h$ is not metrically regular
around $(0_4,0)$. Further, to prove that~\eqref{eq:leb} does
not hold, consider the sequence $\left\{ x_{k}\right\} $ defined
by
$
x_{k}=\left(0,0,\frac{1}{k},\frac{1}{k}\right)
$.
We see that $\left\{ x_{k}\right\} \to0_4$ and
\[
\mathrm{dist}(x_{k},\Omega)=\|x_{k}\|=\frac{\sqrt{2}}{k}=\mathcal{O}(k^{-1}).
\]
Since $\|h(x_{k})\|=\frac{\sqrt{26}}{k^{2}}=\mathcal{O}(k^{-2}),$
we conclude that~\eqref{eq:leb} does not hold.

Consider the polynomial function $\psi(x):=\frac{1}{2}\|h(x)\|^{2}$
of degree $4$, which satisfies $\psi^{-1}(0)=0_4$. It follows from
Fact~\ref{thm:Loj_exponent} that there exist some constants $\beta>0$
and $r>0$ such that
\[
\frac{1}{2}\|h(x)\|^{2}=\psi(x)\geq\beta\|x\|^{(4-1)^{4}+1}=\beta\|x\|^{82},\quad\forall x\in\mathbb{B}(0_4,r).
\]
This implies that $h$ is H\"{o}lder metrically subregular of order $\delta=\frac{1}{41}$
around $(0_4,0)$.
Nonetheless, the order $\frac{1}{41}$ given by Fact~\ref{thm:Loj_exponent} can be improved by using the theory of $2$-regularity: the function $h$ turns out to be $2$-regular at $0_4$, which implies by~\cite[Theorem~4]{izmailov2001} that~\eqref{eq:errorBound} holds with $\delta=\frac{1}{2}$ (see also~\cite[Remark~7]{izmailov2001}). Recall that a twice differentiable mapping $h:\mathbb{R}^m\to\mathbb{R}^n$ is said to be \emph{2-regular} at the point $\overline{x}$ if the range of $\psi_2(z)$ is $\mathbb{R}^n$ for all $z\in T_2\setminus\{0\}$, where $\psi_2:\mathbb{R}^m\to \mathbb{R}^{n\times m}$ is defined for $z\in\mathbb{R}^m$ by
\begin{gather*}
\psi_2(z):=\nabla h(\overline{x})^T+D^2 P h(\overline{x})(z,\cdot),\\
T_2:=\left\{z\in\mathbb{R}^m\mid \nabla h(\overline{x})^T z=0_n\text{ and } D^2 Ph(\overline{x})(z,z)=0_n\right\},
\end{gather*}
$P$ is the projector in $\mathbb{R}^n$ onto the complementary subspace to the range of $\nabla h(\overline{x})^T$, and $D^2$ stands for the second-order (Fr\'echet) derivative.

Indeed, for any $z\in\mathbb{R}^4$, one has $\nabla h(0_4)^T z=(z_1 + 10z_2, \sqrt{5}(z_3 - z_4),0,0)^T$, so the range of $\nabla h(0_4)^T$ is $Y_1=\mathbb{R}^2\times\{0_2\}$, whose complementary subspace is $Y_2=\{0_2\}\times\mathbb{R}^2$. Then, $T_2=\{(-10t,t,0,0)^T, t\in\mathbb{R}\}$ and for each $z\in T_2\setminus \{0_4\}$, one has
\begin{align*}
\psi_2(z)&=\begin{bmatrix}
1 & 10 & 0 & 0\\
0 & 0 & \sqrt{5} & -\sqrt{5}\\
0 & 2t & -4t & 0\\
-20\sqrt{10}t & 0 & 0 & 20\sqrt{10}t
\end{bmatrix},
\end{align*}
which is full-rank for all $t\neq 0$. Therefore, the range of $\psi_2(z)$ is equal to $\mathbb{R}^4$ for all $z\in T_2\setminus \{0_4\}$, and the function $h$ is $2$-regular at $0_4$.\qede
\end{example}

There are many examples of smooth functions that are H\"older metrically subregular of order~$\delta$ around some zero of the function and whose gradient is not full row rank at that point,
cf. \cite{izmailov2001,izmailov2002}. Nonetheless, the following result restricts the possible values of $\delta$: if $x^*$ is an isolated solution in~$\Omega$
(i.e., the function is \emph{H\"older strongly metrically subregular} at $x^*$, cf.~\cite{Mordukhovich2015,CDK18}), and $\nabla h$ is Lipschitz continuous around $x^*$ then one must have $\delta\in{]0,1/2]}$ if $\delta\neq 1$. In fact, only H\"older continuity of $\nabla h$ is needed. Recall that a function $g:\mathbb{R}^m\to\mathbb{R}^n$ is said to be \emph{H\"older continuous} of order $\upsilon\in{]0,1]}$ with constant $L>0$ around some point $x^*\in\mathbb{R}^m$ whenever there exist a positive constant $r$ such that
$$\|g(x)-g(y)\|\leq L\|x-y\|^\upsilon,\quad\forall x,y\in\mathbb{B}(x^*,r).$$
When $\upsilon=1$, $g$ is said to be Lipschitz continuous with constant $L$ around $x^*$.

\begin{proposition}\label{prop:about_delta}
Let $h:\mathbb{R}^m\to\mathbb{R}^n$ be a continuously differentiable function which is H\"older metrically subregular of order $\delta$ around some isolated solution $x^*\in\Omega=\{x\in\mathbb{R}^m:h(x)=0\}$. Assume further that $\nabla h$ is H\"older continuous around $x^*$ of order $\upsilon\in{]0,1]}$ and that $\nabla h(x^*)$ is not full row rank. Then, it holds that $\delta\in\left]0,\frac{1}{1+\upsilon}\right]$.
\end{proposition}

\begin{proof}
Because of the H\"older continuity assumption and the mean value theorem, there are some positive constants $L$ and $r$ such that, for all $x,y \in \mathbb{B}(x^*,r)$, it holds
\begin{align}\label{eq:Holder}
\|h(y)-h(&x)-\nabla h(x)^T(y-x)\|\nonumber\\
&=\left\|\int_0^1\nabla h(x+t(y-x))^T(y-x)dt-\nabla h(x)^T(y-x)\right\|\nonumber\\
&\leq\|y-x\|\int_0^1\left\|\nabla h(x+t(y-x))-\nabla h(x)\right\|dt\nonumber\\
&\leq L\|y-x\|^{1+\upsilon}\int_0^1 t^\upsilon dt=\frac{L}{1+\upsilon}\|y-x\|^{1+\upsilon}.
\end{align}
By using the fact that $x^*$ is an isolated solution, it is possible to make $r$ smaller if needed so that~\eqref{eq:errorBound} holds and
$$\|x-x^*\|=\mathrm{dist}(x,\Omega),\quad \forall x\in\mathbb{B}(x^*,r).$$
Since $\nabla h(x^*)$ is not full row rank, there exists some $z\neq 0$ such that $\nabla h(x^*)^Tz=0$. Consider now the points
$$w_k:=x^*+\frac{r}{k\|z\|}z,\quad\text{with }k=1,2,\ldots.$$
Observe that
$$\nabla h(x^*)^T(w_k-x^*)=\frac{r}{k\|z\|}\nabla h(x^*)^T z=0.$$
As $w_k\in\mathbb{B}(x^*,r)$ for all $k$,  we deduce
\begin{align*}
\beta\|w_k-x^*\|&=\beta\mathrm{dist}(w_k,\Omega)\leq\|h(w_k)\|^\delta\\
&=\|h(w_k)-h(x^*)-\nabla h(x^*)(w_k-x^*)\|^\delta\\
&\leq \frac{L^\delta}{(1+\upsilon)^\delta}\|w_k-x^*\|^{(1+\upsilon)\delta}.
\end{align*}
Thus, we get
$$\|w_k-x^*\|^{(1+\upsilon)\delta-1}\geq \frac{\beta(1+\upsilon)^\delta}{L^\delta},$$
which implies that $\delta\leq \frac{1}{1+\upsilon}$, since $w_k\to x^*$, as claimed.\qed
\end{proof}

The next example shows that the full rank assumption in Proposition~\ref{prop:about_delta} is not redundant,
and that the upper bound $\delta\leq \frac{1}{1+\upsilon}$ can be attained.

\begin{example}\label{ex:new}
Consider the continuously differentiable functions $h,\widehat{h}:\mathbb{R}\to\mathbb{R}$ given for $x\in\mathbb{R}$ by $h(x):=\frac{3}{4}\sqrt[3]{x^4}$ and $\widehat{h}(x):=\frac{3}{4}\sqrt[3]{x^4}+x$, whose solution sets are $\Omega=\{0\}$ and $\widehat{\Omega}=\left\{-\frac{64}{27},0\right\}$, respectively. Let $x^*:=0\in\Omega\cap\widehat{\Omega}$. Then,
$h'(x)=\sqrt[3]{x}$ and $\widehat{h}'(x)=\sqrt[3]{x}+1$, which are both H\"older continuous around $x^*$ of order $\upsilon=\widehat{\upsilon}=\frac{1}{3}$.
Observe that $h'(0)=0$ while $\widehat{h}'(0)=1$. Hence, it follows that $\widehat{h}$ is (H\"older) metrically subregular around $x^*$ of order
$\widehat{\delta}:=1>\frac{1}{1+\widehat{\upsilon}}$, while it is easy to check that  $h$ is H\"older metrically subregular around $x^*$ of order
$\delta:=\frac{3}{4}=\frac{1}{1+\upsilon}$.\qede
\end{example}

\section{Local convergence of the Levenberg--Marquardt method}\label{sec:Local}

In this section, to solve a nonlinear system of the form~\eqref{eq:nonequa},
we consider an adaptive Levenberg--Marquardt method and investigate its local convergence near a solution. Specifically, we consider the following Levenberg--Marquardt
algorithm.

\RestyleAlgo{boxruled}
\begin{algorithm}[ht!]
\DontPrintSemicolon \KwIn{ $x_{0}\in\mathbb{R}^{m}$,
$\eta>0$, $\xi_0\in[\xi_{\min},\xi_{\max}]$, $\omega_0 \in [\omega_{\min},\omega_{\max}]$, with $\xi_{\min}+\omega_{\min}>0$;} \Begin{ $k:=0$;~$\mu_{0}:=\xi_{0}\|h(x_{0})\|^{\eta}+\omega_0\|\nabla h(x_{0})h(x_{0})\|^{\eta};$\;
\While {$\|h(x_{k+1})\|>0$}{ solve the linear system~\eqref{eq:LM_direction}
to specify the direction $d_{k}$;\; $x_{k+1}=x_{k}+d_{k}$;\;update
 $\xi_{k}\in[\xi_{\min},\xi_{\max}]$, $\omega_k\in[\omega_{\min},\omega_{\max}]$ and compute $\mu_{k}$ with~\eqref{eq:muk};\; } } \caption{(Levenberg--Marquardt method with Adaptive Regularisation)\label{ALGORITHM-1:llm}}
\end{algorithm}

In order to prove the local convergence of algorithm~\ref{ALGORITHM-1:llm} to some solution
$x^{*}\in\Omega$, we assume throughout the paper that the next two conditions hold:
\begin{description}
\item [\mylabel{(A1)}{(A1)}] There exists some constants $r\in{]0,1[}$, $\lambda>0$, $\beta>0$ and $\delta\in{]0,1]}$ such that the function $h$ is continuously differentiable and Lipschitz continuous with constant $\lambda$ on $\mathbb{B}(x^*,r)$, and is H\"{o}lder
metrically subregular of order $\delta$ around $(x^{*},0)$;
that is,~\eqref{eq:errorBound}
holds.
\item [\mylabel{(A2)}{(A2)}] $\nabla h$ is H\"older continuous of order $\upsilon\in{]0,1]}$ with constant $L>0$ on $\mathbb{B}(x^*,r)$.
\end{description}

Note that from \ref{(A1)}-\ref{(A2)} and the mean value
theorem, see~\eqref{eq:Holder}, it holds
\begin{equation}
\left\|h(y)-h(x)-\nabla h(x)^{T}(y-x)\right\|\leq \frac{L}{1+\upsilon}\|y-x\|^{1+\upsilon},\quad\forall x,y\in\mathbb{B}(x^{*},r).\label{gradientInequality}
\end{equation}

Let us define the constants
$$\widetilde{r}:=\left\{\begin{array}{ll} \frac{r}{2}, &\text{if }\xi_{\min}>0,\\
\min\left\{\frac{r}{2},\left(\frac{\beta^2(1+\upsilon)^{2\delta}}{2^\delta L^{2\delta}}\right)^\frac{1}{2\delta(1+\upsilon)-2}\right\}, & \text{otherwise},
\end{array}\right.$$
and
$$
\varpi:=\left\{\begin{array}{ll} 1, &\text{if }\xi_{\min}>0,\\
2-\delta, & \text{otherwise}.
\end{array}\right.$$

We begin our study with an analysis inspired by~\cite{alefeld_rate_2001}, \cite{fischer2002local} and~\cite{Guo2015}. The following result provides a bound for the norm of the direction
$d_{k}$ based on the distance of the current iteration $x_{k}$ to
the solution set~$\Omega$. This will be useful later for deducing
the rate of convergence of~\ref{ALGORITHM-1:llm}.
\begin{proposition}
\label{prop:dk_estimate} If $\xi_{\min}=0$, assume that $\delta>\frac{1}{1+\upsilon}$. Let $x_{k}\not\in\Omega$ be an iteration
generated by~\ref{ALGORITHM-1:llm} with $\eta\in{]0,2\delta(1+\upsilon)/\varpi[}$.
Then, if $x_{k}\in\mathbb{B}(x^{*},\widetilde{r})$, the direction $d_{k}$ given
by~\eqref{eq:LM_direction} satisfies
\begin{equation}
\|d_{k}\|\leq\beta_{1}\mathrm{dist}\left(x_{k},\Omega\right)^{\delta_{1}},\label{dk_estimate}
\end{equation}
where $\delta_1:=\min\left\{ 1+\upsilon-\frac{\eta\varpi}{2\delta},\,1\right\}$ and
$$\beta_1:=\left\{\begin{array}{ll}\sqrt{L^{2}{(1+\upsilon)}^{-2}\xi_{\min}^{-1}\beta^{-\frac{\eta}{\delta}}+1}, &\text{if }\xi_{\min}>0,\\
\sqrt{L^2 4^\eta\omega_{\min}^{-1}{(1+\upsilon)}^{-2}\beta^{-\frac{2\eta}{\delta}}+1}, & \text{otherwise}.
\end{array}\right.
$$
\end{proposition}
\begin{proof}
For all $k$, we will denote by $\overline{x}_{k}$ a vector in $\Omega$
such that $\|x_{k}-\overline{x}_{k}\|=\mathrm{dist}(x_{k},\Omega)$.
Since $x_{k}\in\mathbb{B}(x^{*},r/2),$ we have
\[
\|\overline{x}_{k}-x^{*}\|\leq\|\overline{x}_{k}-x_{k}\|+\|x_{k}-x^{*}\|\leq2\|x_{k}-x^{*}\|\leq r,
\]
which implies $\overline{x}{}_{k}\in\mathbb{B}(x^{*},r)$. Further,
\begin{equation}
\|\overline{x}_{k}-x_{k}\|=\mathrm{dist}(x_{k},\Omega)\leq\|x_{k}-x^{*}\|\leq\frac{r}{2}<1.\label{eq:norm_small_1}
\end{equation}

Observe that $\phi_{k}$ is strongly convex and the global
minimiser of $\phi_{k}$ is given by~\eqref{eq:LM_direction}. Then,
we have
\begin{equation}
\phi_{k}(d_{k})\leq\phi_{k}(\overline{x}_{k}-x_{k}).\label{gamma_inequality}
\end{equation}
From the definition of $\phi_{k}$ in~\eqref{eq:def_phi_k}, by~\eqref{gradientInequality} and~\eqref{gamma_inequality}, we deduce
\begin{equation}
\begin{split}\|d_{k}\|^{2} & \leq\frac{1}{\mu_{k}}\phi_{k}(d_{k})\leq\frac{1}{\mu_{k}}\phi_{k}(\overline{x}_{k}-x_{k})\\
 & =\frac{1}{\mu_{k}}\left(\|\nabla h(x_{k})^{T}(\overline{x}_{k}-x_{k})+h(x_{k})\|^{2}+\mu_{k}\|\overline{x}_{k}-x_{k}\|^{2}\right)\\
 & =\frac{1}{\mu_{k}}\left(\|\nabla h(x_{k})^{T}(\overline{x}_{k}-x_{k})+h(x_{k})-h(\overline{x}_{k})\|^{2}+\mu_{k}\|\overline{x}_{k}-x_{k}\|^{2}\right)\\
 & \leq\frac{1}{\mu_{k}}\left(\frac{L^{2}}{{(1+\upsilon)}^2}\|\overline{x}_{k}-x_{k}\|^{2(1+\upsilon)}+\mu_{k}\|\overline{x}_{k}-x_{k}\|^{2}\right).
\end{split}
\label{eq:dk_bounded}
\end{equation}

Let us assume first that $\xi_{\min}>0$. It follows from the definition of $\mu_{k}$ in~\eqref{eq:muk} and~\eqref{eq:errorBound}
that
\begin{align*}
\mu_{k}&\geq\xi_{k}\|h(x_{k})\|^{\eta}\geq\xi_{\min}\|h(x_{k})\|^{\eta}\\
&\geq\xi_{\min}\beta^{\frac{\eta}{\delta}}\mathrm{dist}(x_{k},\Omega)^{\frac{\eta}{\delta}}=\xi_{\min}\beta^{\frac{\eta}{\delta}}\|\overline{x}_{k}-x_{k}\|^{\frac{\eta}{\delta}},
\end{align*}
leading to
\begin{align*}
\|d_{k}\|^{2}&\leq \frac{L^{2}}{{(1+\upsilon)}^2}\xi_{\min}^{-1}\beta^{-\frac{\eta}{\delta}}\|\overline{x}_{k}-x_{k}\|^{2(1+\upsilon)-\frac{\eta}{\delta}}+\|\overline{x}_{k}-x_{k}\|^{2}\\
&\leq\left(\frac{L^{2}}{{(1+\upsilon)}^2}\xi_{\min}^{-1}\beta^{-\frac{\eta}{\delta}}+1\right)\|\overline{x}_{k}-x_{k}\|^{\mbox{\ensuremath{\min}}\left\{ 2(1+\upsilon)-\frac{\eta}{\delta},\,2\right\}},
\end{align*}
and this completes the proof of~\eqref{dk_estimate} for the case $\xi_{\min}>0$.

Let us consider now the case where $\xi_{\min}=0$, assuming then $\delta>\frac{1}{1+\upsilon}$. By~\eqref{gradientInequality},~\eqref{eq:errorBound}  and the Cauchy--Schwarz inequality, we have
\begin{align}
\frac{L^{2}}{{(1+\upsilon)}^2}\mathrm{dist}(x_{k},\Omega)^{2(1+\upsilon)}&\geq\left\|h(x_k)+\nabla h(x_k)^T(\overline{x}_k-x_k)\right\|^2\nonumber\\
&=\|h(x_k)\|^2+2(\overline{x}_k-x_k)^T\nabla h(x_k)h(x_k)\nonumber\\
&\quad+\left\|\nabla h(x_k)^T(\overline{x}_k-x_k)\right\|^2\nonumber\\
&\geq \beta^{\frac{2}{\delta}}\mathrm{dist}(x_{k},\Omega)^{\frac{2}{\delta}}-2\|\overline{x}_k-x_k\|\|\nabla h(x_k)h(x_k)\|.\label{eq:remark}
\end{align}
Thus, since $x_k\not\in\Omega$, we deduce
$$\|\nabla h(x_k)h(x_k)\|\geq \frac{\beta^{\frac{2}{\delta}}}{2}\mathrm{dist}(x_{k},\Omega)^{\frac{2}{\delta}-1}-\frac{L^2}{2{(1+\upsilon)}^2}
\mathrm{dist}(x_{k},\Omega)^{1+2\upsilon}.$$
Since $\delta>\frac{1}{1+\upsilon}$, we have
\begin{equation}\label{eq:bound_dist}
\begin{split}
\frac{L^2}{2{(1+\upsilon)}^2}
\mathrm{dist}(x_{k},\Omega)^{1+2\upsilon-\left(\frac{2}{\delta}-1\right)}&\leq \frac{L^2}{2{(1+\upsilon)}^2}\left\|x_k-x^*\right\|^{2\left(1+\upsilon-\frac{1}{\delta}\right)}\\
&\leq\frac{L^2}{2{(1+\upsilon)}^2}\widetilde{r}^{\,2\left(1+\upsilon-\frac{1}{\delta}\right)}\leq \frac{\beta^{\frac{2}{\delta}}}{4},
\end{split}
\end{equation}
and therefore
\begin{equation*}
\|\nabla h(x_k)h(x_k)\|\geq\frac{\beta^{\frac{2}{\delta}}}{4}\mathrm{dist}(x_{k},\Omega)^{\frac{2}{\delta}-1}.
\end{equation*}
This, together with the definition of $\mu_{k}$ in~\eqref{eq:muk}, implies
$$
\mu_k\geq \omega_k\|\nabla h(x_k)h(x_k)\|^\eta\geq\frac{\omega_{\min}\beta^{\frac{2\eta}{\delta}}}{4^\eta}
\left\|\overline{x}_k-x_k\right\|^{\left(\frac{2}{\delta}-1\right)\eta}.
$$
Using~\eqref{eq:dk_bounded}, we obtain
\begin{align*}
\|d_k\|^2&\leq \frac{L^2 4^\eta}{\omega_{\min}{(1+\upsilon)}^2\beta^{\frac{2\eta}{\delta}}}\left\|\overline{x}_k-x_k\right\|^{2(1+\upsilon)
-\left(\frac{2}{\delta}-1\right)\eta}+\|\overline{x}_{k}-x_{k}\|^{2}\\
&\leq\left(\frac{L^2 4^\eta}{\omega_{\min}{(1+\upsilon)}^2\beta^{\frac{2\eta}{\delta}}}+1\right)\left\|\overline{x}_k-x_k\right\|^{\min\left\{{2(1+\upsilon)
-\left(\frac{2}{\delta}-1\right)\eta},2\right\}},
\end{align*}
which completes the proof.\qed
\end{proof}
\begin{remark}
If $\delta>\frac{1}{1+\upsilon}$, by~\eqref{eq:remark}, we have that $\nabla h(x_k) h(x_k)=0$ implies $x_k\in\Omega$ whenever $x_k$ is sufficiently close to~$x^*$.
\end{remark}

The next result provides an upper bound for the distance of $x_{k+1}$
to the solution set $\Omega$ based on the distance of $x_{k}$ to $\Omega$.
\begin{proposition}
\label{distance_estimate} If $\xi_{\min}=0$, assume that $\delta>\frac{1}{1+\upsilon}$. Let $x_{k}\not\in\Omega$ and $x_{k+1}$ be two
consecutive iterations generated by~\ref{ALGORITHM-1:llm} with $\eta\in{]0,2\delta(1+\upsilon)/\varpi[}$.
Then, if $x_{k},x_{k+1}\in\mathbb{B}(x^{*},\widetilde{r})$, we have
\begin{equation}
\mathrm{dist}(x_{k+1},\Omega)\leq\beta_{2}\mathrm{dist}(x_{k},\Omega)^{\delta_{2}},\label{eq:distineq}
\end{equation}
where $\beta_2$ is a positive constant and
\begin{equation}\label{eq:delta_2}
\delta_2:=\min\left\{ (1+\upsilon)\delta,\,\left(1+\frac{\eta}{2}\right)\delta,\,(1+\upsilon)\left(\delta+
\delta\upsilon-\frac{\eta\varpi}{2}\right)\right\}.
\end{equation}
\end{proposition}

\begin{proof}
Let $\overline{x}_{k}\in\Omega$ be such that $\|x_{k}-\overline{x}_{k}\|=\mathrm{dist}(x_{k},\Omega)$.
From the definition of $\phi_{k}$ in~\eqref{eq:def_phi_k} and the reasoning in~\eqref{eq:dk_bounded},
we obtain
\begin{align*}
\|\nabla h(x_{k})^{T}d_{k}+h(x_{k})\|^{2} & \leq\phi_{k}(d_{k})\leq \frac{L^{2}}{{(1+\upsilon)}^2}\|\overline{x}_{k}-x_{k}\|^{2(1+\upsilon)}+\mu_{k}\|\overline{x}_{k}-x_{k}\|^{2}.
\end{align*}
It follows from \ref{(A1)} that there exists some constant $\widehat{L}$ such that $\|\nabla h(x)\|\leq \widehat{L}$ for all $x\in\mathbb{B}(x^*,r)$. Then, by the definition of $\mu_{k}$ in~\eqref{eq:muk} and the Lipschitz continuity of $h$, we have
that
\begin{equation}\label{eq:muk_bound}
\begin{split}
\mu_{k} & =\xi_{k}\|h(x_{k})\|^{\eta}+\omega_k\|\nabla h(x_{k})h(x_{k})\|^{\eta} \\ &\leq\xi_{\max}\|h(x_{k})\|^{\eta}+\omega_{\max}\widehat{L}^{\eta}\|h(x_{k})\|^{\eta}\\
 & =\left(\xi_{\max}+\omega_{\max}\widehat{L}^{\eta}\right)\|h(x_{k})-h(\overline{x}_{k})\|^{\eta}\\
& \leq\left(\xi_{\max}+\omega_{\max}\widehat{L}^{\eta}\right)\lambda^{\eta}\|x_{k}-\overline{x}_{k}\|^{\eta},
\end{split}
\end{equation}
which implies, thanks to~\eqref{eq:norm_small_1},
\begin{align*}
{\left\|\nabla h(x_{k})^{T}d_{k}+h(x_{k})\right\|}^{2} & \leq \frac{L^{2}}{{(1+\upsilon)}^2}\|\overline{x}_{k}-x_{k}\|^{2(1+\upsilon)}\\
&\quad+\left(\xi_{\max}+\omega_{\max}\widehat{L}^{\eta}\right)\lambda^{\eta}\|x_{k}-\overline{x}_{k}\|^{2+\eta}\\
 & \leq\left(\frac{L^{2}}{{(1+\upsilon)}^2}+\lambda^{\eta}\xi_{\max}+\widehat{L}^{\eta}\lambda^{\eta}\omega_{\max}\right)\|x_{k}-\overline{x}_{k}\|^{\zeta},
\end{align*}
where $\zeta:=\min\left\{ 2(1+\upsilon),\,2+\eta\right\} $. By~\eqref{eq:errorBound},~\eqref{gradientInequality}, the latter inequality and Proposition~\ref{prop:dk_estimate}, we get
{\allowdisplaybreaks
\begin{align*}
\left(\beta\mathrm{dist}(x_{k+1},\Omega)\right)^{\frac{1}{\delta}} & \leq\|h(x_{k}+d_{k})\|\\
 & \leq\left\|\nabla h(x_{k})^{T}d_{k}+h(x_{k})\right\|\\
 &\quad+\left\|h(x_{k}+d_{k})-h(x_{k})-\nabla h(x_{k})^{T}d_{k}\right\|\\
 & \leq\left\|\nabla h(x_{k})^{T}d_{k}+h(x_{k})\right\|+\frac{L}{1+\upsilon}\|d_{k}\|^{1+\upsilon}\\
 & \leq\sqrt{L^{2}{(1+\upsilon)}^{-2}+\lambda^{\eta}\xi_{\max}+\widehat{L}^{\eta}\lambda^{\eta}\omega_{\max}}\|x_{k}-\overline{x}_{k}\|^{\frac{\zeta}{2}}\\
 &\quad+\frac{L\beta_{1}^{1+\upsilon}}{{1+\upsilon}}\mathrm{dist}\left(x_{k},\Omega\right)^{(1+\upsilon)\delta_{1}}\\
 & \leq\widehat{\beta}_2\mathrm{dist}(x_{k},\Omega)^{\widehat{\delta_{2}}},
\end{align*}}%
where 
\begin{align*}
\widehat{\delta_{2}}&:=\min\left\{ \frac{\zeta}{2},\,(1+\upsilon)\delta_{1}\right\} =\mbox{\ensuremath{\min}}\left\{ 1+\upsilon,\,1+\frac{\eta}{2},\,(1+\upsilon)\left(1+\upsilon-\frac{\eta\varpi}{2\delta}\right)\right\},\\ \widehat{\beta}_2&:=\sqrt{L^{2}{(1+\upsilon)}^{-2}+\lambda^{\eta}\xi_{\max}+\widehat{L}^{\eta}\lambda^{\eta}\omega_{\max}}+L\beta_{1}^{1+\upsilon}{(1+\upsilon)}^{-1}.
\end{align*}
Therefore,
\[
\mathrm{dist}(x_{k+1},\Omega)\leq\beta_{2}\mathrm{dist}\left(x_{k},\Omega\right)^{\delta\widehat{\delta_{2}}}
=\beta_{2}\mathrm{dist}\left(x_{k},\Omega\right)^{\delta_{2}},
\]
with $\delta_2$ given by~\eqref{eq:delta_2} and $\beta_{2}:=\frac{1}{\beta}\widehat{\beta}_2^{\delta}$,
giving the result.\qed
\end{proof}

The following proposition gives a different value of the exponent in~\eqref{eq:distineq}.
\begin{proposition}\label{distance_estimate2}
Assume that $\delta> \frac{1}{1+\upsilon}$. Let $x_{k}\not\in\Omega$ and $x_{k+1}$ be two
consecutive iterations generated by~\ref{ALGORITHM-1:llm} with $\eta\in{]0,2\delta(1+\upsilon)/\varpi[}$
and such that $x_{k},x_{k+1}\in\mathbb{B}(x^{*},\widetilde{r})$.
Then, there exists a positive constant $\beta_3$ such that
\begin{equation}
\mathrm{dist}(x_{k+1},\Omega)\leq\beta_{3}\mathrm{dist}(x_{k},\Omega)^{\delta_{3}},\label{eq:distineq2}
\end{equation}
where
\begin{equation}\label{eq:delta_3}\small
\delta_3:=
\min\left\{\frac{(1+\eta)\delta}{2-\delta},\frac{(1+\upsilon)\delta}{2-\delta},\frac{\eta\delta+(1+\upsilon)\delta
-\frac{\eta\varpi}{2}}{2-\delta},\frac{{(1+\upsilon)}^2\delta
-(1+\upsilon)\frac{\eta\varpi}{2}}{2-\delta}\right\}.
\end{equation}
\end{proposition}
\begin{proof}
Let $\overline{x}_k,\overline{x}_{k+1}\in\Omega$ be such that $\|x_{k}-\overline{x}_{k}\|=\mathrm{dist}(x_{k},\Omega)$ and $\|x_{k+1}-\overline{x}_{k+1}\|=\mathrm{dist}(x_{k+1},\Omega)$. Assume that $x_{k+1}\not\in\Omega$ (otherwise, the inequality trivially holds). By~\eqref{gradientInequality}, we have
\begin{align*}
\left\|h(x_{k+1})+\nabla h(x_{k+1})^T\left(\overline{x}_{k+1}-x_{k+1}\right)\right\|^2&\leq \frac{L^2}{{(1+\upsilon)}^2}\left\|\overline{x}_{k+1}-x_{k+1}\right\|^{2(1+\upsilon)}\\
&=\frac{L^2}{{(1+\upsilon)}^2}\mathrm{dist}(x_{k+1},\Omega)^{2(1+\upsilon)}.
\end{align*}
Thus, by the Cauchy--Schwarz inequality and~\eqref{eq:errorBound}, we get
\begin{align*}
-\left\|\nabla h(x_{k+1})h(x_{k+1})\right\|&\mathrm{dist}(x_{k+1},\Omega)\\
&\leq h(x_{k+1})^T\nabla h(x_{k+1})^T\left(\overline{x}_{k+1}-x_{k+1}\right)\\
&\leq \frac{L^2}{2{(1+\upsilon)}^2}\mathrm{dist}(x_{k+1},\Omega)^{2(1+\upsilon)}-\frac{1}{2}\|h(x_{k+1})\|^2\\
&\quad-\frac{1}{2}\|\nabla h(x_{k+1})^T\left(\overline{x}_{k+1}-x_{k+1}\right)\|^2\\
&\leq\frac{L^2}{2{(1+\upsilon)}^2}\mathrm{dist}(x_{k+1},\Omega)^{2(1+\upsilon)}-\frac{\beta^{\frac{2}{\delta}}}{2}\mathrm{dist}(x_{k+1},\Omega)^{\frac{2}{\delta}},
\end{align*}
that is,
\begin{equation}\label{eq:dist_bound}
\begin{split}
\frac{\beta^{\frac{2}{\delta}}}{2}\mathrm{dist}(x_{k+1},\Omega)^{\frac{2}{\delta}}&-\frac{L^2}{2{(1+\upsilon)}^2}\mathrm{dist}(x_{k+1},\Omega)^{2(1+\upsilon)}\\
&\leq\left\|\nabla h(x_{k+1})h(x_{k+1})\right\|\mathrm{dist}(x_{k+1},\Omega).
\end{split}
\end{equation}
Now, by~\eqref{eq:LM_direction}, we have
\begin{equation}\label{eq:norm_k_1}
\begin{split}
\|\nabla &h(x_{k+1})h(x_{k+1})\|\\
&=\left\|\nabla h(x_{k+1})h(x_{k+1})-\nabla h(x_k)\left(h(x_k)+\nabla h(x_k)^Td_k\right)-\mu_kd_k\right\|\\
&\leq \left\|\nabla h(x_{k+1})-\nabla h(x_{k})\right\|\|h(x_{k+1})\|\\
&\quad+\left\|\nabla h(x_k)\right\|\left\|h(x_{k+1})-h(x_k)-\nabla h(x_k)^T(x_{k+1}-x_k)\right\|+\mu_k\|d_k\|\\
&\leq L\|d_k\|^\upsilon\|h(x_{k+1})\|+\frac{L}{1+\upsilon}\|\nabla h(x_k)\|\|d_k\|^{1+\upsilon}+\mu_k\|d_k\|.
\end{split}
\end{equation}
By~\ref{(A1)} and Proposition~\ref{prop:dk_estimate}, it holds,
\begin{align*}
\|h(x_{k+1})\|&=\|h(x_{k+1})-h(\overline{x}_k)\|\leq \lambda\|x_{k+1}-\overline{x}_k\|\\
&\leq \lambda\left(\|x_{k+1}-x_k\|+\|x_k-\overline{x}_k\|\right)\\
&\leq \lambda\left(\beta_{1}\mathrm{dist}\left(x_{k},\Omega\right)^{\delta_{1}}+\mathrm{dist}\left(x_{k},\Omega\right)\right)\\
&\leq \lambda(\beta_{1}+1)\mathrm{dist}\left(x_{k},\Omega\right)^{\delta_{1}}.
\end{align*}
It follows from  \ref{(A1)} that there exists some constant $\widehat{L}$ such that $\|\nabla h(x)\|\leq \widehat{L}$ for all $x\in\mathbb{B}(x^*,r)$. Then, by the definition of $\mu_{k}$ in~\eqref{eq:muk} and~\ref{(A1)}, we get~\eqref{eq:muk_bound}.
Hence, by~\eqref{eq:norm_k_1} and Proposition~\ref{prop:dk_estimate}, we deduce
\begin{align*}
\left\|\nabla h(x_{k+1})h(x_{k+1})\right\|&\leq L\lambda\beta_{1}^\upsilon(\beta_1+1)\mathrm{dist}\left(x_{k},\Omega\right)^{\delta_{1}(1+\upsilon)}\\
&\quad+\frac{\widehat{L}L\beta_{1}^{1+\upsilon}}{1+\upsilon}\mathrm{dist}\left(x_{k},\Omega\right)^{\delta_{1}(1+\upsilon)}\\
&\quad+\left(\xi_{\max}+\omega_{\max}\widehat{L}^\eta\right)\lambda^{\eta}\beta_{1}\mathrm{dist}\left(x_{k},\Omega\right)^{\eta+\delta_{1}}\\
&\leq \widehat{\beta}_3\mathrm{dist}\left(x_{k},\Omega\right)^{\widehat{\delta}_{3}},
\end{align*}
where $\widehat{\beta}_3:=L\lambda\beta_{1}^\upsilon(\beta_1+1)+\widehat{L}L\beta_{1}^{1+\upsilon}{(1+\upsilon)}^{-1}
+\left(\xi_{\max}+\omega_{\max}\widehat{L}^\eta\right)\lambda^{\eta}\beta_{1}$ and $\widehat{\delta}_{3}:=\min\{\eta+\delta_1,\delta_1(1+\upsilon)\}$. Therefore, by~\eqref{eq:dist_bound},
\begin{equation}\label{eq:dist_1}
\begin{split}
\frac{\beta^{\frac{2}{\delta}}}{2}\mathrm{dist}(x_{k+1},\Omega)^{\frac{2}{\delta}}&-\frac{L^2}{2{(1+\upsilon)}^{2}}\mathrm{dist}(x_{k+1},\Omega)^{2(1+\upsilon)}\\
&\leq\widehat{\beta}_3\mathrm{dist}\left(x_{k},\Omega\right)^{\widehat{\delta}_{3}}\mathrm{dist}(x_{k+1},\Omega).
\end{split}
\end{equation}
Since $\delta>\frac{1}{1+\upsilon}$, we have by~\eqref{eq:bound_dist} that
$$\frac{L^2}{2{(1+\upsilon)}^{2}}\mathrm{dist}(x_{k+1},\Omega)^{2(1+\upsilon)-\frac{2}{\delta}}\leq\frac{\beta^{\frac{2}{\delta}}}{4}.$$
Finally, by~\eqref{eq:dist_1}, we deduce
$$\frac{\beta^{\frac{2}{\delta}}}{4}\mathrm{dist}(x_{k+1},\Omega)^{\frac{2}{\delta}-1}\leq
\widehat{\beta}_3\mathrm{dist}\left(x_{k},\Omega\right)^{\widehat{\delta}_{3}},
$$
whence,
$$\mathrm{dist}(x_{k+1},\Omega)\leq
\beta_3\mathrm{dist}\left(x_{k},\Omega\right)^{\delta_{3}},
$$
where $\beta_3:=\frac{4\widehat{\beta}_3}{\beta^{\frac{2}{\delta}}}$ and
$\delta_3:=\frac{\widehat{\delta}_3\delta}{2-\delta}$. Since the expression for $\delta_3$ coincides with~\eqref{eq:delta_3}, the proof is complete.\qed
\end{proof}

\begin{remark}\label{rem:requirements}
(i) The bounds given by~\eqref{eq:distineq} and~\eqref{eq:distineq2} are usually employed
to analyse the rate of convergence of the sequence $\left\{ x_{k}\right\}$
generated by~\ref{ALGORITHM-1:llm}. Observe that the values of $\delta_2$ and $\delta_3$ when $\xi_{\min}>0$ are greater or equal than their respective values when $\xi_{\min}=0$. A larger value of $\delta_{2}$ or $\delta_3$ would serve us to derive a better rate of convergence. To deduce a convergence result from Proposition~\ref{distance_estimate}, one needs to have $\delta_2> 1$. This holds if and only if $\delta>\frac{1}{1+\upsilon}$ and $\eta\in\left]\frac{2}{\delta}-2,\frac{1}{\varpi}\left(2\delta(1+\upsilon)-\frac{2}{1+\upsilon}\right)\right[$, which imposes an additional requirement on the value of $\delta$ (to have a nonempty interval). For instance, when $\upsilon=1$, one must have $\delta>\frac{-1+\sqrt{33}}{8}$ if $\xi_{\min}>0$ and $\delta>\frac{-5+\sqrt{57}}{4}$ if $\xi_{\min}=0$.
On the other hand, to guarantee that $\delta_3> 1$, a stronger requirement would be needed, namely, $\delta>\frac{2}{2+\upsilon}\geq\frac{2}{3}$ and $\eta\in\left]\frac{2}{\delta}-2,
\frac{1}{\varpi}\left(2\delta(1+\upsilon)-\frac{4-2\delta}{1+\upsilon}\right)\right[$. Nonetheless, it is important to observe that if $\delta=1$ one has that $\delta_3=1+\upsilon$ when $\eta\in[\upsilon,2\upsilon/\varpi]$, while $\delta_2=1+\upsilon$ only if $\eta=2\upsilon$ and $\varpi=1$. Therefore, if $\upsilon=\delta=1$, we can derive from Proposition~\ref{distance_estimate2} the quadratic convergence of the sequence for $\eta\in[1,2]$, which can only be guaranteed for $\eta=2$ by Proposition~\ref{distance_estimate}. In Figure~\ref{fig:comparison_deltas}, we plot the values of $\delta_2$ in Proposition~\ref{distance_estimate} and $\delta_3$ in Proposition~\ref{distance_estimate2} when $\upsilon=1$ and $\xi_{\min}>0$.

\begin{figure}[ht!]
\centering%
\subfigure{\label{fig:AnotherFig}{\includegraphics[width=.54\textwidth]{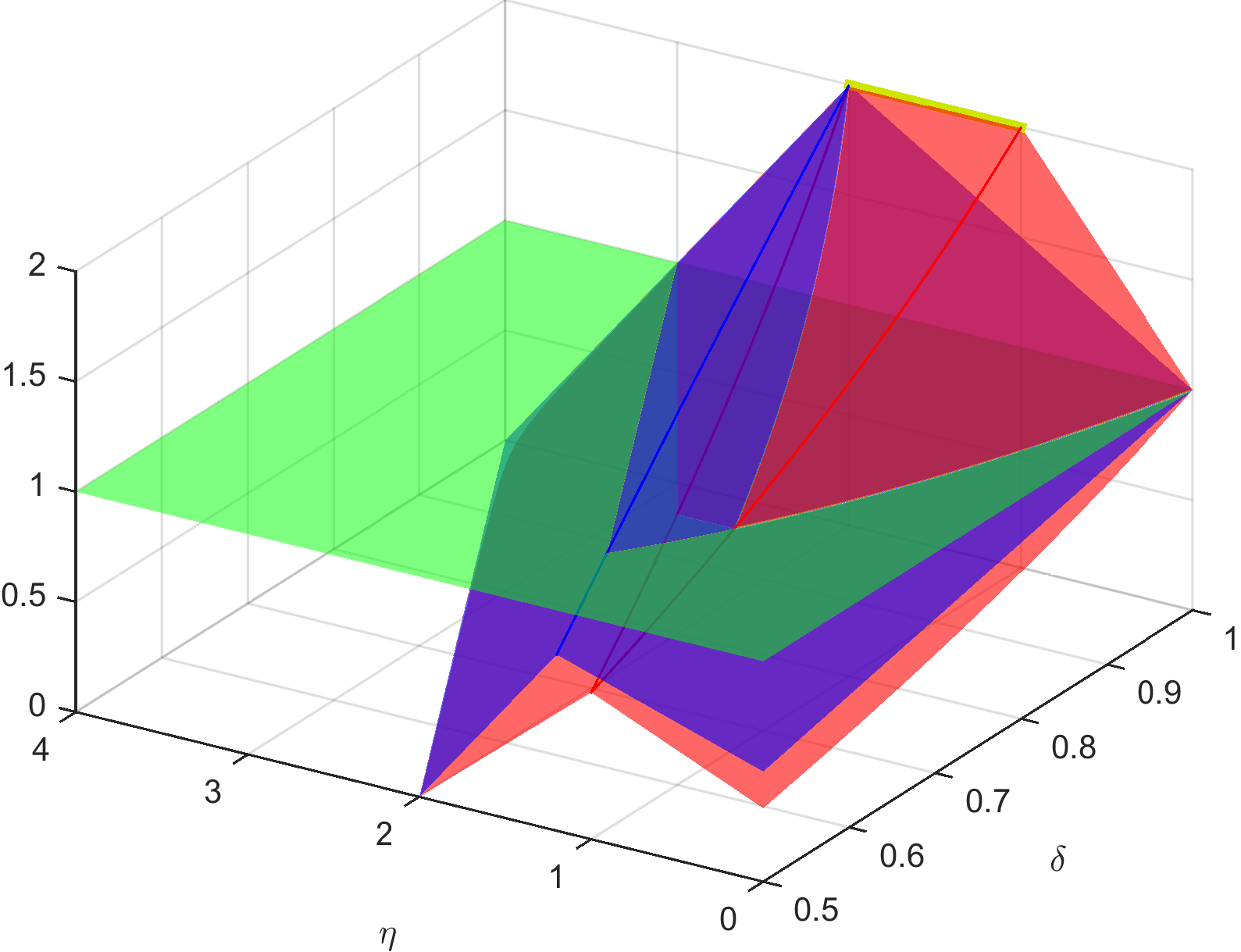}}}\quad
\subfigure{{\includegraphics[width=.43\textwidth]{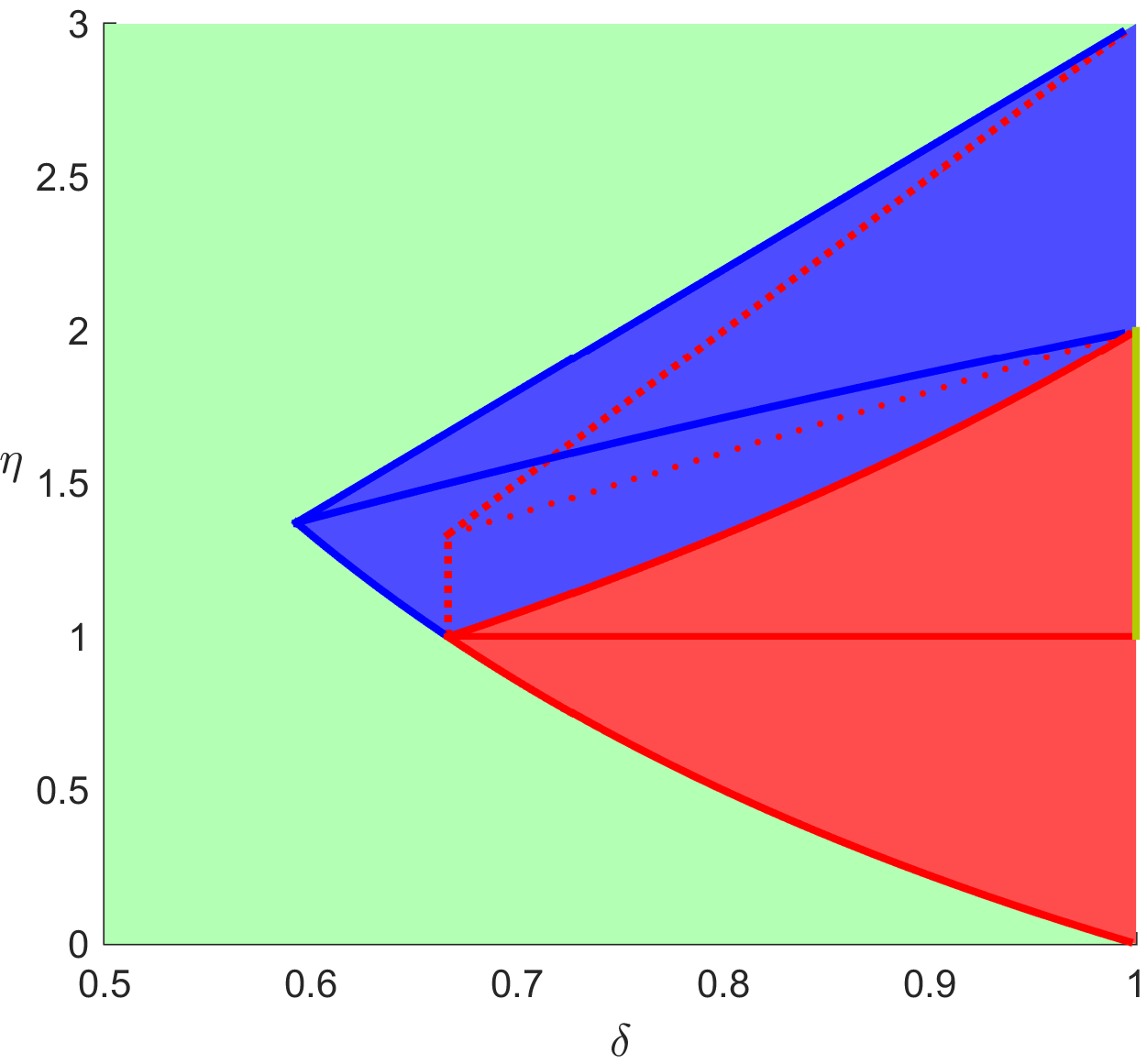}}\label{fig:AnotherFig2}}

\caption{For $\upsilon=1$, $\xi_{\min}>0$, $\delta\in\left[\frac{1}{2},1\right]$ and $\eta\in[0,4\delta]$, plot of $\delta_{2}=\min\left\{ 2\delta,\delta+\frac{\delta\eta}{2},\,4\delta-\eta\right\}$ (in blue) and $\delta_3=\min\left\{\frac{4\delta-\eta}{2-\delta},\frac{(\eta+1)\delta}{2-\delta},
\frac{2\delta}{2-\delta}\right\}$ (in red).}\label{fig:comparison_deltas}
\end{figure}
\noindent (ii) The values of $\delta_2$ and $\delta_3$ are maximised when $\eta=\frac{2\upsilon\delta(2+\upsilon)}{\delta+\varpi(1+\upsilon)}$ and $\eta\in\left[\upsilon,\frac{2\upsilon\delta}{\varpi}\right]$, respectively, in which case $\delta_{2}=  \delta+{\frac {\upsilon\delta^2 \left( 2+\upsilon \right) }{\delta+\varpi(1+\upsilon)}}$ and $\delta_3=\frac{(1+\upsilon)\delta}{2-\delta}$, having then $\delta_2\leq\delta_3$.

\end{remark}

\begin{remark}\label{rem:warning}
In light of Proposition~\ref{prop:about_delta}, the extent of the results that can be derived from Propositions~\ref{distance_estimate} and~\ref{distance_estimate2} is rather reduced when $x^*$ is an isolated solution and $\nabla h(x^*)$ is not full rank, since it imposes $\delta \le \frac{1}{1+\upsilon}$. Note that the function $F_S$ given as an example in~\cite[Section~5]{Guo2015} is H\"older metrically subregular of order \mbox{$\delta=\frac{5}{6}>0.5$}, but $\nabla F_S$ is not Lipschitz continuous around any zero of the function, so it does not satisfy~\ref{(A2)} for $\upsilon=1$ (and, therefore, it does not satisfy~\cite[Assumption~4.1]{Guo2015} either).
However, with the additional assumption that the \L{}ojasiewicz gradient inequality~\eqref{eq:Lojasiewicz_Gradient_Inequality} holds, we
will obtain local convergence for all $\delta \in {]0,1]}$ (see Theorem~\ref{thm:convergence_of_dist(x,Omega)_general}).
\end{remark}

Next, we proceed to derive the main result of this section from Propositions~\ref{distance_estimate} and~\ref{distance_estimate2}, where we provide a region from which the parameter $\eta$ must be chosen so that superlinear convergence is guaranteed. Recall that a sequence $\left\lbrace z_k \right\rbrace $ is said to converge \emph{superlinearly} to $z^*$ with order $q>1$  if $z_k$ converges to $z^*$ and there exists $K > 0$ such that $\|z_{k+1}-z^*\| \leq K \|z_k -z^*\|^{q}$ for all $k$ sufficiently large.

\begin{theorem}
\label{thm:dist_con_sup} Assume that $\delta>\frac{1}{1+\upsilon}$ and $\eta\in\left ]\frac{2}{\delta}-2,\frac{1}{\varpi}\left(2\delta(1+\upsilon)-\frac{2}{1+\upsilon}\right)\right[$.
Then, there exists some $\overline{r}>0$  such that, for every
sequence $\left\{ x_{k}\right\} $ generated by~\ref{ALGORITHM-1:llm} with $x_{0}\in\mathbb{B}(x^{*},\overline{r})$,
one has that $\left\{ \mathrm{dist}(x_{k},\Omega)\right\} $ is superlinearly convergent to
$0$ with order $\delta_{2}$ given by~\eqref{eq:delta_2}.
Further, the sequence $\left\{ x_{k}\right\} $
converges to a solution $\overline{x}\in \Omega \cap\mathbb{B}(x^{*},\widetilde{r})$, and if $\eta\leq \frac{2\upsilon\delta}{\varpi}$, its rate of convergence is also superlinear with order $\delta_{2}$.
Moreover, if $\delta>\frac{2}{2+\upsilon}$ and
$\eta<
\frac{1}{\varpi}\left(2(1+\upsilon)\delta-\frac{4-2\delta}{1+\upsilon}\right)$, all the latter holds with order  $\delta_3$ given by~\eqref{eq:delta_3}.
\end{theorem}

\begin{proof}
We assume that $x_k\not\in\Omega$ for all $k$ (otherwise, the statement trivially holds).
Let $\delta_1$, $\beta_1$ be defined as in Proposition~\ref{prop:dk_estimate} and $\delta_2$, $\beta_2$  be defined as in the proof of Proposition~\ref{distance_estimate}. Since $\delta_2>1$, we have that $\delta_1\delta_2^i>i$ for all $i$ sufficiently large. As $\sum_{i=1}^\infty \left(\frac{1}{2}\right)^i=1$, we deduce that
\begin{equation}\label{def:sigma}
\sigma:=\sum_{i=1}^\infty\left(\frac{1}{2}\right)^{\delta_1\delta_2^i}<\infty.
\end{equation}
Define
\[
\overline{r}:=\mbox{min}\left\{ \frac{1}{2}\left(\beta_{2}\right)^{\frac{-1}{\delta_{2}-1}},\left(\frac{\widetilde{r}}{1+\beta_{1}+2^{\delta_{1}}\beta_{1}\sigma}\right)^{\frac{1}{\delta_{1}}}\right\} .
\]
Note that $\overline{r}\in{]0,\widetilde{r}[}$, because $\widetilde{r}\in{]0,1[}$ and $\delta_1\leq 1$.

Pick any $x_{0}\in\mathbb{B}(x^{*},\overline{r})$ and let $\left\{ x_{k}\right\} $ be an infinite sequence generated by~\ref{ALGORITHM-1:llm}. First, we will show by induction that $x_k\in\mathbb{B}(x^*,\widetilde{r})$. It follows from $\overline{r}<1$
and~\eqref{dk_estimate} that
\begin{equation}\label{eq:x1_estimate}
\begin{split}
\|x_{1}-x^{*}\| & =\|x_{0}+d_{0}-x^{*}\|\leq\|x_{0}-x^{*}\|+\|d_{0}\|
 \leq\overline{r}+\beta_{1}\mathrm{dist}\left(x_{0},\Omega\right)^{\delta_{1}}\\
 &\leq\overline{r}^{\delta_{1}}+\beta_{1}\|x_{0}-x^{*}\|^{\delta_{1}} \leq(1+\beta_{1})\overline{r}^{\delta_{1}}\leq \widetilde{r}.
\end{split}
\end{equation}
Let us assume now that
$x_{i}\in\mathbb{B}(x^{*},\widetilde{r})$ for $i=1,2,\dots,k$. Then, from Proposition
\ref{distance_estimate} and the definition of $\overline{r}$, we
have
\begin{align*}
\mathrm{dist}(x_{i},\Omega) & \leq\beta_{2}\mathrm{dist}(x_{i-1},\Omega)^{\delta_{2}}\leq\beta_{2}^{1+\delta_{2}}\mathrm{dist}(x_{i-2},\Omega)^{\delta_{2}^{2}}\\
&\leq\ldots\leq\beta_{2}^{\sum_{j=0}^{i-1}\delta_{2}^{j}}\mathrm{dist}(x_{0},\Omega)^{\delta_{2}^{i}}\\
 & \leq\beta_{2}^{\sum_{j=0}^{i-1}\delta_{2}^{j}}\|x_{0}-x^{*}\|^{\delta_{2}^{i}}=\beta_{2}^{\frac{\delta_{2}^{i}-1}{\delta_{2}-1}}\|x_{0}-x^{*}\|^{\delta_{2}^{i}}\\
 &\leq\left(\frac{1}{2\overline{r}}\right)^{\delta_{2}^{i}-1}\overline{r}^{\delta_{2}^{i}}=2\overline{r}\left(\frac{1}{2}\right)^{\delta_{2}^{i}},
\end{align*}
which yields
\begin{equation}
\mathrm{dist}(x_{i},\Omega)^{\delta_{1}}\leq\left(2\overline{r}\right)^{\delta_{1}}\left(\frac{1}{2}\right)^{\delta_{1}\delta_{2}^{i}}.\label{eq:ball_constraint}
\end{equation}
The latter inequality, together with~\eqref{dk_estimate},~\eqref{def:sigma} and~\eqref{eq:x1_estimate}, implies
\begin{align*}
\|x_{k+1}-x^{*}\| & \leq\|x_{1}-x^{*}\|+\sum_{i=1}^{k}\|d_{i}\|\leq(1+\beta_{1})\overline{r}^{\delta_{1}}+\beta_{1}\sum_{i=1}^{k}\mathrm{dist}(x_{i},\Omega)^{\delta_{1}}\\
 & \leq(1+\beta_{1})\overline{r}^{\delta_{1}}+\beta_{1}\left(2\overline{r}\right)^{\delta_{1}}\sum_{i=1}^{k}\left(\frac{1}{2}\right)^{\delta_{1}\delta_{2}^{i}}\\
 &<(1+\beta_{1})\overline{r}^{\delta_{1}}+\beta_{1}\left(2\overline{r}\right)^{\delta_{1}}\sum_{i=1}^{\infty}\left(\frac{1}{2}\right)^{\delta_{1}\delta_{2}^{i}}\\
 & =(1+\beta_{1})\overline{r}^{\delta_{1}}+\beta_{1}\left(2\overline{r}\right)^{\delta_{1}}\sigma=\left(1+\beta_{1}+2^{\delta_{1}}\beta_{1}\sigma\right)\overline{r}^{\delta_{1}}\leq\widetilde{r},
\end{align*}
which completes the induction. Thus, we have shown that $x_{k}\in\mathbb{B}(x^{*},\widetilde{r})$ for all $k$, as claimed.

From Proposition~\ref{distance_estimate}, we obtain that $\left\{ \mathrm{dist}(x_{k},\Omega)\right\} $
is superlinearly convergent to $0$. Further, it follows from~\eqref{dk_estimate} and~\eqref{eq:ball_constraint}
that
\[
\sum_{i=1}^{\infty}\|d_{i}\|\leq\beta_{1}\sum_{i=1}^{\infty}\mathrm{dist}(x_{i},\Omega)^{\delta_{1}}\leq\beta_{1}\sigma\left(2\overline{r}\right)^{\delta_{1}}<\infty.
\]
Denoting by $s_{k}:=\sum_{i=1}^{k}\|d_{i}\|$, we have that $\left\{ s_{k}\right\} $
is a Cauchy sequence. Then, for any $k,p\in\mathbb{N}\cup\{0\}$,
we have
\begin{equation}\label{eq:bound_xkp}
\begin{split}
\|x_{k+p}-x_{k}\| &\leq \|d_{k+p-1}\|+\|x_{k+p-1}-x_{k}\|\\
&\leq\ldots\leq\sum_{i=k}^{k+p-1}\|d_{i}\|= s_{k+p-1}-s_{k-1},
\end{split}
\end{equation}
which implies that $\left\{ x_{k}\right\} $
is also a Cauchy sequence.
Thus, the sequence $\left\{ x_{k}\right\} $ converges to some~$\overline{x}$. Since $x_{k}\in\mathbb{B}(x^{*},\widetilde{r})$
for all $k$ and $\left\{ \mathrm{dist}(x_{k},\Omega)\right\} $ converges
to $0$, we have $\overline{x}\in\Omega\cap\mathbb{B}(x^{*},\widetilde{r})$.

Further, if $\eta\leq \frac{2\upsilon\delta}{\varpi}$ we have $\delta_1=1$ in Proposition~\ref{prop:dk_estimate}, and by letting $p\to\infty$ in~\eqref{eq:bound_xkp}, we deduce
\begin{align*}
\|\overline{x}-x_k\|&\leq\sum_{i=k}^\infty\|d_i\|\leq\beta_1\sum_{i=k}^{\infty}\mathrm{dist}(x_{i},\Omega).
\end{align*}
Since $\{\mathrm{dist}(x_{k},\Omega)\}$ is superlinearly convergent to zero, for all $k$ sufficiently large, it holds that $\mathrm{dist}(x_{k+1},\Omega)\leq\frac{1}{2}\mathrm{dist}(x_{k},\Omega)$. Therefore, for $k$ sufficiently large, we have
\begin{align*}
\|x_k-\overline{x}\|&\leq \beta_1\sum_{i=k}^{\infty}\frac{1}{2^{i-k}}\mathrm{dist}(x_{k},\Omega)\leq 2\beta_1\mathrm{dist}(x_{k},\Omega)
\leq 2\beta_1\beta_{2}\mathrm{dist}(x_{k-1},\Omega)^{\delta_{2}}\\
&\leq 2\beta_1\beta_{2}\|x_{k-1}-\overline{x}\|^{\delta_{2}},
\end{align*}
which proves the superlinear convergence of $x_k$ to $\overline{x}$ with order $\delta_{2}$.

Finally, the last assertion follows by the same argumentation, using $\delta_3$, $\beta_3$ and Proposition~\ref{distance_estimate2} instead of $\delta_2$, $\beta_2$ and Proposition~\ref{distance_estimate}, respectively.\qed
\end{proof}

\begin{remark}\label{rem:ZhuLin}
Our results above generalise the results in~\cite{Guo2015,ZhuLin16}, not only because in these works they assume $\nabla h$ to be Lipschitz continuous (i.e., $\upsilon=1$), but also because the parameter~$\mu_k$ considered by these authors is equal to $\xi\|h(x_k)\|^\eta$. Furthermore, in their convergence results, cf.~\cite[Theorem~4.1 and Theorem~4.2]{Guo2015} and~\cite[Theorem~2.1 and Theorem~2.2]{ZhuLin16}, the authors assume
$\delta>\max\left\{\frac{2}{3},\frac{2+\eta}{5}\right\}$ and $\delta>\max\left\{\frac{\sqrt{8\eta+1}+4\eta+1}{16},\frac{2}{2+\eta},\frac{1}{2+\eta}+\frac{\eta}{4},\frac{\eta+1}{4}\right\}>\frac{\sqrt{5}-1}{2}$,
respectively, which both entail $\delta>\frac{-1+\sqrt{33}}{8}$, so we have slightly improved the lower bound on $\delta$ for the superlinear convergence in Theorem~\ref{thm:dist_con_sup}.
\end{remark}

As a direct consequence of Theorem~\ref{thm:dist_con_sup}, whenever $\delta=\upsilon=1$ and $\eta\in[1,2]$, we can derive quadratic convergence of the sequence generated by~\ref{ALGORITHM-1:llm}.

\begin{corollary}
\label{cor:dist_con_sup} Assume that $\delta=1$ and $\eta\in{]0, 2\upsilon]}$. Then, there exists $\overline{r}>0$  such that for every
sequence~$\left\{ x_{k}\right\} $ generated by~\ref{ALGORITHM-1:llm} with $x_{0}\in\mathbb{B}(x^{*},\overline{r})$,
one has that $\left\{ \mathrm{dist}(x_{k},\Omega)\right\} $ is superlinearly convergent to~$0$ with order
$$\delta_3=\left\{\begin{array}{ll}
1+\eta,&\text{ if }\eta\leq\upsilon,\\
1+\upsilon,&\text{ if }\eta\geq\upsilon.
\end{array}\right.
$$
Moreover, the sequence $\left\{ x_{k}\right\} $
converges superlinearly with order $\delta_3$ to a solution $\overline{x}\in \Omega \cap\mathbb{B}(x^{*},\widetilde{r})$.
Therefore, when $\upsilon=1$ and $\eta\in[1,2]$, the sequence $\left\{ \mathrm{dist}(x_{k},\Omega)\right\} $ is quadratically convergent to $0$, and the sequence $\left\{ x_{k}\right\}$
converges quadratically to a solution $\overline{x}\in \Omega \cap\mathbb{B}(x^{*},\widetilde{r})$.
\end{corollary}
\begin{remark}
In particular, Corollary~\ref{cor:dist_con_sup} generalizes~\cite[Theorem~3.7]{MJ07}, where the authors prove quadratic convergence of the sequence $\{x_k\}$ by assuming $\delta=\upsilon=1$, and where the parameters in~\eqref{eq:muk} are chosen as $\eta=1$, $\xi_k=\theta\in{[0,1]}$ and $\omega_k=1-\theta$, for all $k$.
\end{remark}

\begin{example}[Example~\ref{ex:new} revisited]
Let $h$ and $\widehat{h}$ be the functions defined in Example~\ref{ex:new}. The function $h$ does not satisfy the assumptions of Theorem~\ref{thm:dist_con_sup}, since $\delta=\frac{1}{1+\upsilon}$. On the other hand, if $\widehat{\eta}\in\left]0,\frac{7}{6}\right[$ and the starting point $x_0$ is chosen sufficiently close to~$0$, Theorem~\ref{thm:dist_con_sup} proves for the function $\widehat{h}$ the superlinear convergence of the sequence generated by~\ref{ALGORITHM-1:llm} to $0$ with order
$$\delta_3=\left\{\begin{array}{ll}
1+\widehat{\eta},&\text{ if }0<\widehat{\eta}<\frac{1}{3},\\
\frac{4}{3},&\text{ if }\frac{1}{3}\leq \widehat{\eta}\leq\frac{2}{3},\\
\frac{4}{3}\left(\frac{4}{3}-\frac{\widehat{\eta}}{2}\right),&\text{ if }\frac{2}{3}< \widehat{\eta}<\frac{7}{6}.
\end{array}\right.
$$
Note that, since the solution is locally unique, the additional assumption $\widehat{\eta}\leq 2\widehat{\upsilon}\hspace{1pt}\widehat{\delta}=\frac{2}{3}$ is not needed.
The order of convergence $\delta_3$ is thus maximised when $\widehat{\eta}\in\left[\frac{1}{3},\frac{2}{3}\right]$.\qede
\end{example}

The question of whether the sequence $\left\{ \mathrm{dist}(x_{k},\Omega)\right\} $ converges to $0$ when $\delta$ does not satisfy the requirements commented in Remark~\ref{rem:requirements}(i) remains open. However, with the additional assumption that $\psi$ satisfies the \L{}ojasiewicz gradient inequality (which holds for real analytic functions), we can prove that the sequences $\left\{ \mathrm{dist}(x_{k},\Omega)\right\} $ and $\left\{ \psi(x_k)\right\} $ converge to $0$ for all $\delta \in (0,1]$ as long as the parameter $\eta$ is sufficiently small, and we can also provide a rate of convergence that depends on the exponent of the \L{}ojasiewicz gradient inequality. This is the subject of the next subsection.

\subsection{Convergence analysis under the \L{}ojasiewicz gradient inequality }\label{Loja_added}
To prove our convergence result, we make use of the following two lemmas.

\begin{lemma}
\label{lem:rate_convergence} Let~$\left\{ s_{k}\right\} $ be a nonnegative real sequence
and let~$\alpha,\vartheta$ be some nonnegative constants.
Suppose that $s_{k}\to0$ and that the sequence satisfies
\begin{equation*}
s_{k}^{\alpha}\leq \vartheta(s_{k}-s_{k+1}),\quad\text{for all }k\text{ sufficiently large.}\label{eq:ineq_seq}
\end{equation*}
Then
\begin{enumerate}
\item if~$\alpha=0$, the sequence~$\left\{ s_{k}\right\} $ converges
to~$0$ in a finite number of steps;
\item if~$\alpha\in\left ]0,1\right]$, the sequence~$\left\{ s_{k}\right\} $
converges linearly to~$0$ with rate~$1-\frac{1}{\vartheta}$;
\item if~$\alpha>1$, there exists~$\varsigma>0$ such that
\[
s_{k}\leq\varsigma k^{-\frac{1}{\alpha-1}},\quad\text{for all }k\text{ sufficiently large.}
\]
\end{enumerate}
\end{lemma}
\begin{proof}
See~\cite[Lemma~1]{artacho_accelerating_2015}.\qed
\end{proof}

\begin{lemma}\label{Svaiter_lemma}
 The sequence $\{x_{k}\}$ generated by~\ref{ALGORITHM-1:llm} satisfies
$$ \|d_{k}\|\leq\frac{1}{2\sqrt{\mu_{k}}}\|h(x_{k})\|,$$
and
\begin{equation*}\label{h(x_k)_estimate}
\begin{split}
\|h(x_{k+1})\|^{2} &\leq\|h(x_{k})\|^{2}+d_{k}^{T}\nabla h(x_{k})h(x_{k})\\
&\quad +\|d_k\|^2\left(\frac{L^2}{{(1+\upsilon)}^2} \|d_k\|^{2\upsilon}+\frac{2L}{1+\upsilon} \|h(x_{k})\| \|d_k\|^{\upsilon-1}-\mu_k \right).
\end{split}
\end{equation*}
\end{lemma}
\begin{proof}
This result is a straightforward modification of~\cite[Theorem 2.5 and Lemma 2.3]{karas2015algebraic}, using~\eqref{eq:Holder} instead of the Lipschitz continuity of $\nabla h$.\qed
\end{proof}

In our second main result of this paper, under the additional assumption that the \L{}ojasiewicz gradient inequality holds, we prove the convergence to $0$ of the sequences $\left\{ \mathrm{dist}(x_{k},\Omega)\right\} $ and $\left\{ \psi(x_k)\right\} $.

\begin{theorem}
\label{thm:convergence_of_dist(x,Omega)_general} Suppose that $\psi$ satisfies the \L{}ojasiewicz gradient inequality~\eqref{eq:Lojasiewicz_Gradient_Inequality} with
exponent $\theta \in {]0,1[}$. Let
\begin{equation}\label{eq:tilde_theta}
\chi:=\left\{\begin{array}{ll}
1,&\text{if }\left(\omega_{\min}=0\right)\text{ or }\left(\xi_{\min}> 0\text{ and }\theta\leq\frac{1}{2}\right),\\
2\theta,&\text{otherwise.}
\end{array}\right.
\end{equation}
Then, if $\eta\in{\left]0,\min\left\{\frac{2\upsilon}{\chi (1+\upsilon)},\frac{2(1-\theta)}{\chi }\right\}\right[}$, there exist some positive constants $s$ and $\overline{s}$ such that, for every $x_{0}\in\mathbb{B}(x^{*},s)$ and every sequence $\{x_{k}\}$ generated by~\ref{ALGORITHM-1:llm}, one has $\{x_k\}\subset\mathbb{B}(x^*,\overline{s})$ and the two sequences $\{\psi(x_{k})\}$ and $\{\mathrm{dist}(x_{k},\Omega)\}$ converge to $0$. Moreover, the following holds:
\begin{enumerate}
\item if $\theta\in\left ]0,\frac{1}{2}\right]$, the sequences~$\{\psi(x_{k})\}$ and $\{\mathrm{dist}(x_{k},\Omega)\}$ converge linearly to~$0$;
\item if $\theta \in\left ]\frac{1}{2},1\right[$, there exist some positive constants $\varsigma_1$ and $\varsigma_2$ such that, for all large~$k$,
\begin{gather*}
\psi(x_{k})\leq \varsigma_1 k^{-\frac{1}{2\theta-1}}\quad\text{and}\quad
\mathrm{dist}(x_{k},\Omega)\leq \varsigma_2 k^{-\frac{\delta}{2(2\theta-1)}}.
\end{gather*}
\end{enumerate}
\end{theorem}
\begin{proof}
The proof has three key parts.

In the first part of the proof, we will set the values of $s$ and $\overline{s}$.
Let $\varepsilon>0$ and $\kappa>0$ be such that~\eqref{eq:Lojasiewicz_Gradient_Inequality} holds. Thus, one has
\begin{equation}\label{eq:lower_bound_second_term}
\|\nabla h(x) h(x)\|=\|\nabla\psi(x)\|\geq\frac{1}{\kappa}\psi(x)^\theta=\frac{1}{2^\theta\kappa}\|h(x)\|^{2\theta},\quad \forall x\in\mathbb{B}(x^*,\varepsilon).
\end{equation}
Let $\overline{s}:=\min\{r,\varepsilon\}>0$.
Then, by Assumption~\ref{(A1)}, there exists some positive constant $ M $ such that
\begin{equation}\label{eq:bound_mu}
\left\|\nabla h(x_k)\nabla h(x_k)^T\right\|+\mu_k\leq M ,\quad\text{whenever }x_k\in\mathbb{B}(x^*,\overline{s}).
\end{equation}
Since $\eta\in{\left]0,\frac{2\upsilon}{\chi (1+\upsilon)}\right[}$, it is possible to make $\overline{s}$ smaller if needed to ensure, for all $x\in\mathbb{B}(x^*,\overline{s})$, that
\begin{equation}\label{eq:choose_s}
\left(\xi_{\min}+\frac{\omega_{\min}}{2^{\theta\eta}\kappa^\eta}\right)\|h(x)\|^{\eta \chi} \geq\left(\frac{2+\sqrt{5}}{2^\upsilon(1+\upsilon)}L\right)^{\frac{2}{1+\upsilon}}\|h(x)\|^{\frac{2\upsilon}{1+\upsilon}}.
\end{equation}
For all $x\in\mathbb{B}(x^*,\overline{s})$, one has by the Lipschitz continuity of $h$ that
\begin{equation}\label{eq:bound_psi}
\psi(x)=\frac{1}{2}\|h(x)-h(x^*)\|^2\leq\frac{\lambda^2}{2}\|x-x^*\|^2\leq\frac{\lambda^2}{2}\|x-x^*\|,
\end{equation}
since $\overline{s}\leq r<1$.
Let
$$\Delta:=\frac{2^{\theta}\kappa M \lambda^{2\left(1-\theta-\frac{\eta \chi}{2}\right)}}
{\left(1-\theta-\frac{\eta \chi}{2}\right)\left(\xi_{\min}+\frac{\omega_{\min}}{2^{\theta\eta}\kappa^\eta}\right)} \quad \text{and} \quad {s}:= \left( \frac{\overline{s}}{1+\Delta}\right)^{\frac{1}{1-\theta-\frac{\eta \chi}{2}}}.$$
Then, since $\overline{s}<1$ and $\theta+\frac{\eta \chi}{2}\in {]0,1[}$, we have $s\leq \overline{s}$.

In the second part of the proof, we will prove by induction that
\begin{align}
&x_{i}\in\mathbb{B}(x^*,\overline{s}),\quad\text{and}\label{eq:induction1}\\
&\|d_{i-1} \| \leq \frac{2^{1-\frac{\eta \chi}{2}}\kappa M }{\left(1-\theta-\frac{\eta \chi}{2}\right)\left(\xi_{\min}+\frac{\omega_{\min}}{2^{\theta\eta}\kappa^\eta}\right)} \left(\psi(x_{i-1})^{1-\theta-\frac{\eta \chi}{2}}-\psi(x_{i})^{1-\theta-\frac{\eta \chi}{2}}\right),\label{eq:induction}
\end{align}
for all $i=1,2,\ldots$. Pick any $x_{0}\in\mathbb{B}(x^{*},s)$ and let $\{x_k\}$ be the sequence generated by \ref{ALGORITHM-1:llm}.
It follows from Lemma~\ref{Svaiter_lemma} that
\begin{equation}\label{eq:conseq_Svaiter}
\begin{split}
\psi(x_{k+1})&\leq\psi(x_{k})-\frac{1}{2}d_{k}^{T}H_k d_k\\
&\quad+\frac{\|d_{k}\|^{2}}{2\mu_{k}^\upsilon}\left(\frac{L^{2}}{4^\upsilon{(1+\upsilon)}^2}\|h(x_{k})\|^{2\upsilon}+\frac{2^{2-\upsilon}L\mu_{k}^{\frac{1+\upsilon}{2}}}{1+\upsilon}\|h(x_{k})\|^\upsilon
-\mu_{k}^{1+\upsilon}\right),
\end{split}
\end{equation}
for all $k$, where $H_k=\nabla h(x_k)\nabla h(x_k)^T+\mu_k I$, since $d_k=-H_k^{-1}\nabla h(x_{k})h(x_{k})$.
Since $x_0\in\mathbb{B}(x^*,\overline{s})$, we have by~\eqref{eq:lower_bound_second_term}, the definition of $\chi $ in~\eqref{eq:tilde_theta} and~\eqref{eq:choose_s} that
\begin{align}
\mu_{0} &\geq \xi_{\min}\|h(x_0)\|^\eta+\omega_{\min}\|\nabla h(x_0)h(x_0)\|^\eta\nonumber\\
&\geq\xi_{\min}\|h(x_0)\|^\eta+\frac{\omega_{\min}}{2^{\theta\eta}\kappa^\eta}\|h(x_0)\|^{2\theta\eta}\nonumber\\
&\geq\left(\xi_{\min}+\frac{\omega_{\min}}{2^{\theta\eta}\kappa^\eta}\right)\|h(x_0)\|^{\eta \chi}
\geq\left(\frac{2+\sqrt{5}}{2^\upsilon(1+\upsilon)}L\right)^{\frac{2}{1+\upsilon}}
\|h(x_0)\|^{\frac{2\upsilon}{1+\upsilon}},\label{eq:1}
\end{align}
 which implies
 $$\frac{L^{2}}{4^\upsilon{(1+\upsilon)}^2}\|h(x_{0})\|^{2\upsilon}+\frac{2^{2-\upsilon}L\mu_{0}^{\frac{1+\upsilon}{2}}}{1+\upsilon}\|h(x_{0})\|^\upsilon
-\mu_{0}^{1+\upsilon}\leq 0.$$
Therefore, from~\eqref{eq:conseq_Svaiter}, we get
\begin{equation}\label{eq:psi_x1}
\psi(x_{1})\leq\psi(x_{0})-\frac{1}{2}d_{0}^{T}H_0 d_0\leq \psi(x_{0})-\frac{\mu_0}{2}\|d_0\|^2.
\end{equation}
Observe that the convexity of the function $\varphi(t):=-t^{1-\theta-\frac{\eta\chi}{2}}$ with $t>0$ yields
\begin{equation}\label{eq:convexity_varphi}
\psi(x)^{1-\theta-\frac{\eta\chi}{2}}-\psi(y)^{1-\theta-\frac{\eta\chi}{2}}
\geq  \left(1-\theta-\frac{\eta\chi}{2}\right)\psi(x)^{-\theta-\frac{\eta\chi}{2}}\left( \psi(x)-\psi(y)\right),
\end{equation}
for all $x,y\in\mathbb{R}^m\setminus\Omega$. By combining~\eqref{eq:psi_x1} with~\eqref{eq:convexity_varphi}, we deduce
\begin{equation}\label{psi_estimate_1}
\psi(x_0)^{1-\theta-\frac{\eta\chi}{2}}-\psi(x_1)^{1-\theta-\frac{\eta\chi}{2}}\geq \frac{\left(1-\theta-\frac{\eta\chi}{2}\right)\mu_0}{2} \psi(x_0)^{-\theta-\frac{\eta\chi}{2}}\|d_0\|^2
\end{equation}
Since $x_0\in\mathbb{B}(x^*,s)\subseteq\mathbb{B}(x^*,\overline{s})$, we have by~\eqref{eq:bound_mu} that $\|H_0\|\leq M $. Further, by the \L{}ojasiewicz gradient inequality~\eqref{eq:Lojasiewicz_Gradient_Inequality}, it holds
$$
\psi(x_0)^{\theta} \leq \kappa \|\nabla \psi (x_0) \|
 \leq \kappa \|H_0\| \|d_0 \| \leq \kappa  M \|d_0 \|.
$$
From the last inequality, together with~\eqref{psi_estimate_1}, the first inequality in~\eqref{eq:1} and then~\eqref{eq:bound_psi}, we obtain
\begin{align*}
\|d_0 \| &\leq \frac{2\kappa M \psi(x_0)^{\frac{\eta\chi}{2}}}{\left(1-\theta-\frac{\eta\chi}{2}\right)\mu_0} \left(\psi(x_0)^{1-\theta-\frac{\eta\chi}{2}}-\psi(x_{1})^{1-\theta-\frac{\eta\chi}{2}}\right)\\
&\leq \frac{2\kappa M  }{\left(1-\theta-\frac{\eta\chi}{2}\right)\left(\xi_{\min}+\frac{\omega_{\min}}{2^{\theta\eta}\kappa^\eta}\right)2^{\frac{\eta \chi}{2}}} \left(\psi(x_0)^{1-\theta-\frac{\eta\chi}{2}}-\psi(x_{1})^{1-\theta-\frac{\eta\chi}{2}}\right)\\
&\leq\frac{2^{1-\frac{\eta \chi}{2}}\kappa M }{\left(1-\theta-\frac{\eta\chi}{2}\right)\left(\xi_{\min}+\frac{\omega_{\min}}{2^{\theta\eta}\kappa^\eta}\right)} \psi(x_0)^{1-\theta-\frac{\eta \chi}{2}}\leq \Delta\|x_0-x^*\|^{1-\theta-\frac{\eta \chi}{2}},
\end{align*}
which, in particular, proves~\eqref{eq:induction} for $i=1$.
Hence,
\begin{align*}
\|x_1-x^*\| &\leq \|x_0-x^*\|+\|d_0\|
\leq \|x_0-x^*\|+\Delta \|x_0-x^*\|^{1-\theta-\frac{\eta \chi}{2}}\\
&\leq (1+\Delta) \|x_0-x^*\|^{1-\theta-\frac{\eta \chi}{2}}
\leq (1+\Delta) s^{1-\theta-\frac{\eta \chi}{2}} =\overline{s}.
\end{align*}
Therefore, $x_1\in\mathbb{B}(x^*,\overline{s})$. Assume now that~\eqref{eq:induction1}--\eqref{eq:induction} holds for all $i=1,\ldots, k$.
Since $x_k\in\mathbb{B}(x^*,\overline{s})$, by~\eqref{eq:choose_s} and the same argumentation as in~\eqref{eq:1}, we have
$$
\mu_{k} \geq \left(\xi_{\min}+\frac{\omega_{\min}}{2^{\theta\eta}\kappa^\eta}\right)\|h(x_k)\|^{\eta \chi} \geq\left(\frac{2+\sqrt{5}}{2^\upsilon(1+\upsilon)}L\right)^{\frac{2}{1+\upsilon}}\|h(x_k)\|^{\frac{2\upsilon}{1+\upsilon}},$$
 which implies
 $$\frac{L^{2}}{4^\upsilon{(1+\upsilon)}^2}\|h(x_{k})\|^{2\upsilon}+\frac{2^{2-\upsilon}L\mu_{k}^{\frac{1+\upsilon}{2}}}{1+\upsilon}\|h(x_{k})\|^\upsilon
-\mu_{k}^{1+\upsilon}\leq 0.$$
Therefore, by~\eqref{eq:conseq_Svaiter}, we get
\begin{equation}\label{eq:psi_nonincreasing}
\psi(x_{k+1})\leq \psi(x_k)-\frac{1}{2} d_k^T H_k d_k\leq\psi(x_k)-\frac{\mu_k}{2}\|d_k\|^2.
\end{equation}
Combining the latter inequality with~\eqref{eq:convexity_varphi}, we deduce
\begin{equation}\label{psi_estimate_k}
\psi(x_k)^{1-\theta-\frac{\eta\chi}{2}}-\psi(x_{k+1})^{1-\theta-\frac{\eta\chi}{2}}\geq \frac{\left(1-\theta-\frac{\eta\chi}{2}\right)\mu_k}{2} \psi(x_k)^{-\theta-\frac{\eta\chi}{2}}\|d_k\|^2
\end{equation}
Further, since $x_k\in\mathbb{B}(x^*,\overline{s})$, from the \L{}ojasiewicz gradient inequality~\eqref{eq:Lojasiewicz_Gradient_Inequality} and~\eqref{eq:bound_mu}, it holds
$$
\psi(x_k)^{\theta} \leq \kappa \|\nabla \psi (x_k) \|
 \leq \kappa \|H_k\| \|d_k \| \leq \kappa  M \|d_k\|.
$$
From the last inequality and~\eqref{psi_estimate_k}, we deduce
\begin{align*}
\|d_k \| &\leq \frac{2\kappa M \psi(x_k)^{\frac{\eta\chi}{2}}}{\left(1-\theta-\frac{\eta\chi}{2}\right)\mu_k} \left(\psi(x_k)^{1-\theta-\frac{\eta\chi}{2}}-\psi(x_{k+1})^{1-\theta-\frac{\eta\chi}{2}}\right)\\
 &\leq  \frac{2^{1-\frac{\eta \chi}{2}}\kappa M }{\left(1-\theta-\frac{\eta\chi}{2}\right)\left(\xi_{\min}+\frac{\omega_{\min}}{2^{\theta\eta}\kappa^\eta}\right) }\left(\psi(x_k)^{1-\theta-\frac{\eta\chi}{2}}-\psi(x_{k+1})^{1-\theta-\frac{\eta\chi}{2}}\right),
\end{align*}
which proves~\eqref{eq:induction} for $i=k+1$. Hence, by~\eqref{eq:bound_psi}, we have
\begin{align*}
\|x_{k+1}&-x^*\| \\
&\leq \|x_0-x^*\|+\sum_{i=0}^k\|d_i\|\\
&\leq \|x_0-x^*\|\\
&\quad+\frac{2^{1-\frac{\eta \chi}{2}}\kappa M }{(1-\theta-\frac{\eta\chi}{2})\left(\xi_{\min}+\frac{\omega_{\min}}{2^{\theta\eta}\kappa^\eta}\right)}\sum_{i=0}^k  \left(\psi(x_i)^{1-\theta-\frac{\eta \chi}{2}}-\psi(x_{i+1})^{1-\theta-\frac{\eta\chi}{2}}\right)\\
&\leq \|x_0-x^*\|+\frac{2^{1-\frac{\eta \chi}{2}}\kappa M }{(1-\theta-\frac{\eta\chi}{2})\left(\xi_{\min}+\frac{\omega_{\min}}{2^{\theta\eta}\kappa^\eta}\right)}\psi(x_0)^{1-\theta-\frac{\eta \chi}{2}}\\
&\leq (1+\Delta) \|x_0-x^*\|^{1-\theta-\frac{\eta \chi}{2}}
\leq (1+\Delta) s^{1-\theta-\frac{\eta \chi}{2}} =\overline{s},
\end{align*}
which proves~\eqref{eq:induction1} for $i=k+1$. This completes the second part of the proof.

In the third part of the proof, we will finally show the assertions in the statement of the theorem.
From the second part of the proof we know that $x_k\in\mathbb{B}(x^*,\overline{s})$ for all $k$. This, together with~\eqref{eq:bound_mu}, implies that
$\|H_k\|\leq  M $ for all $k$.
Thus,
$$d_k^T H_k d_k=\nabla \psi(x_k)^{T}H_k^{-1}\nabla \psi(x_k)
 \geq \frac{1}{\|H_k\|} \|\nabla \psi(x_k)\|^2 \geq \frac{1}{ M }\|\nabla \psi(x_k)\|^2.$$
Therefore, by~\eqref{eq:psi_nonincreasing}, we have
$$
\psi(x_{k+1})\leq \psi(x_k)-\frac{1}{2 M }\|\nabla \psi(x_k)\|^2.
$$
It follows from the \L{}ojasiewicz gradient inequality~\eqref{eq:Lojasiewicz_Gradient_Inequality} and the last inequality that
$$
\psi(x_{k+1})\leq \psi(x_k)-\frac{1}{2\kappa^2 M }\psi(x_k)^{2\theta}.
$$
This implies that $\{\psi(x_k)\}$ converges to $0$. By applying Lemma~\ref{lem:rate_convergence} with $s_k:=\psi(x_k)$, $\vartheta:=2\kappa^2 M $ and $\alpha:=2\theta$,
we conclude that the rate of convergence depends on $\theta$ as claimed in (i)-(ii). Finally, observe that $\{\mathrm{dist}\left(x_{k},\Omega\right)\}$ converges to $0$ with
the rate stated in (i)-(ii) thanks to the H\"older metric subregularity of the function $h$.\qed
\end{proof}

\begin{remark}~\label{rem:real_analytic}
Observe that every real analytic function satisfies the assumptions of Theorem~\ref{thm:convergence_of_dist(x,Omega)_general}, thanks to Fact~\ref{prop:Lojasiewicz Inequality-1} and the discussion after it in Section~\ref{sec.2}. Therefore, local sublinear convergence of~\ref{ALGORITHM-1:llm} is guaranteed for all $\eta$ sufficiently small (i.e., whenever $\eta<\min\left\{\chi^{-1},2(1-\theta)\chi^{-1}\right\}$). This is the best that we can get with these weak assumptions, as we show in the next example.
\end{remark}

\begin{example}[Example~\ref{ex:new} revisited]
Let $h(x)=\frac{3}{4}\sqrt[3]{x^4}$ be the function considered in Example~\ref{ex:new}. The function $h$ does not satisfy the assumptions of Theorem~\ref{thm:dist_con_sup}, but it verifies the ones of Theorem~\ref{thm:convergence_of_dist(x,Omega)_general}. Indeed, it is straightforward to check that $\psi(x)=\frac{1}{2}|h(x)|^2$ satisfies the \L{}ojasiewicz gradient inequality~\eqref{eq:Lojasiewicz_Gradient_Inequality} with exponent $\theta=\frac{5}{8}$. Since $\theta>\frac{1}{2}$, we can only guarantee the sublinear convergence of the sequence $\{x_k\}$ generated by~\ref{ALGORITHM-1:llm} to $0$ when $\eta\in{\left]0,\frac{1}{2\chi}\right[}={\left]0,\min\left\{\frac{1}{2\chi},\frac{3}{4\chi}\right\}\right[}$. In fact, this is the best convergence rate that we can get. Indeed, a direct computation gives us
\begin{equation}\label{eq:example}
x_{k+1}=\left(1-\frac{\frac{3}{4}x_k^{\frac{2}{3}}}{x_k^{\frac{2}{3}}+
\xi_k\left(\frac{3}{4}\right)^\eta|x_k|^{\frac{4\eta}{3}}+\omega_k\left(\frac{3}{4}\right)^\eta|x_k|^{\frac{5\eta}{3}}}\right)x_k.
\end{equation}
On the one hand, when $\xi_{min}>0$ and~$\eta\in{\left]0,\frac{1}{2}\right[}$, we have $\frac{4\eta}{3}<\frac{2}{3}$. Therefore, it follows from~\eqref{eq:example} and $\xi_k\geq\xi_{\min}>0$ that
$$\lim_{k\to\infty}\left|\frac{x_{k+1}}{x_k}\right|=1,$$
which means that $\{x_k\}$ is sublinearly convergent to $0$. This coincides with what Theorem~\ref{thm:convergence_of_dist(x,Omega)_general} asserts, since $\left]0,\frac{1}{2\chi}\right[=\left]0,\frac{1}{2}\right[$. On the other hand, when $\xi_{\min}=0$ and $\eta\in\left]0,\frac{2}{5}\right[$, sublinear convergence is also obtained from~\eqref{eq:example}, which is exactly what Theorem~\ref{thm:convergence_of_dist(x,Omega)_general} guarantees for all $\eta\in \left]0,\frac{1}{2\chi}\right[=\left]0,\frac{2}{5}\right[$.
\end{example}

\section{Application to biochemical reaction networks}\label{sec.numapp}

In this section, we introduce first a class of nonlinear equations arising in the study of biochemistry, cf.~\cite{Fleming2012}. After that, we compare the performance of \ref{ALGORITHM-1:llm} with various Levenberg--Marquardt algorithms for finding steady states of nonlinear systems of biochemical networks on 20 different real data biological models.

\subsection{Nonlinear systems in biochemical reaction networks}\label{bio}

Consider a biochemical network with $m$ molecular
species and $n$ reversible elementary reactions\footnote{An elementary reaction is a chemical reaction for which no intermediate
molecular species need to be postulated in order to describe the chemical
reaction on a molecular scale.}. We define forward and reverse \emph{stoichiometric matrices}, $F,R\in\mathbb{\mathbb{Z}}_{+}^{m\times n}$,
respectively, where $F_{ij}$ denotes the \emph{stoichiometry}\footnote{Reaction stoichiometry is a quantitative relationship between the
relative quantities of molecular species involved in a single chemical
reaction. } of the $i^{th}$ molecular species in the $j^{th}$ forward reaction
and $R_{ij}$ denotes the stoichiometry of the $i^{th}$ molecular
species in the $j^{th}$ reverse reaction.  We assume that \emph{every
reaction conserves mass}, that is, there exists at least one positive
vector $l\in\mathbb{R}_{++}^{m}$ satisfying $(R-F)^{T}l=0$, cf. \cite{gevorgyan2008detection}.
The matrix $N:=R-F$ represents net reaction stoichiometry and may be viewed as the incidence matrix of a directed hypergraph, see  \cite{Klamt2009}. We assume that there are less molecular species than there are net reactions, that is $m < n$. We assume the
cardinality of each row of $F$ and $R$ is at least one, and the cardinality
of each column of $R-F$ is at least two. The matrices $F$ and $R$ are sparse and the particular sparsity
pattern depends on the particular biochemical network being modeled. Moreover, we also assume that $\text{rank}([F,R])=m$, which is a requirement for kinetic consistency, cf.~\cite{Fleming20161}.

Let $c\in\mathbb{R}_{++}^{m}$ denote a variable vector of molecular
species concentrations. Assuming constant nonnegative elementary kinetic
parameters $k_{f},k_{r}\in\mathbb{R}_{+}^{n}$, we assume \emph{elementary
reaction kinetics} for forward and reverse elementary reaction rates
as $s(k_{f},c):=\exp(\ln(k_{f})+F^{T}\ln(c))$ and $r(k_{r},c):=\exp(\ln(k_{r})+R^{T}\ln(c))$,
respectively, where $\exp(\cdot)$ and $\ln(\cdot)$ denote the respective
componentwise functions, see, e.g.,~\cite{artacho_accelerating_2015,Fleming20161}.
Then, the deterministic dynamical equation for time evolution of molecular
species concentration is given by
\begin{eqnarray}
\frac{dc}{dt} & \equiv & N(s(k_{f},c)-r(k_{r},c)) \label{eq:dcdt2}\\
 & = & N\left(\exp(\ln(k_{f})+F^{T}\ln(c)\right)-\exp\left(\ln(k_{r})+R^{T}\ln(c))\right)=:-f(c).\nonumber
\end{eqnarray}
A vector $c^{*}$ is a \emph{steady state} if and only if it satisfies
\[
f(c^{*})=0.
\]
Note that a vector $c^{*}$ is a steady state of the biochemical system
if and only if
\[
s(k_{f},c^{*})-r(k_{r},c^{*})\in\mathcal{N}(N),
\]
here $\mathcal{N}(N)$ denotes the null space of $N$. Therefore,
the set of steady states $\Omega=\left\{ c\in\mathbb{R}_{++}^{m},\,f(c)=0\right\} $
is unchanged if we replace the matrix $N$ by a matrix $\bar{N}$
with the same null space. Suppose that $\bar{N}\in\mathbb{Z}^{r\times n}$
is the submatrix of $N$ whose rows are linearly independent, then $\mathrm{rank}\left(\bar{N}\right)=\mathrm{rank}(N)\eqqcolon r.$
If one replaces $N$ by $\bar{N}$ and transforms~\eqref{eq:dcdt2}
to logarithmic scale, by letting $x\coloneqq\ln(c)\in\mathbb{R}^{m}$,
$k\coloneqq[\ln(k_{f})^{T},\,\ln(k_{r})^{T}]^{T}\in\mathbb{R}^{2n}$,
then the right-hand side of~\eqref{eq:dcdt2} is equal
to the function
\begin{equation*}
\bar{f}(x):=\left[\bar{N},-\bar{N}\right]\exp\left(k+[F,\,R]^{T}x\right),\label{eq:f(x)}
\end{equation*}
where $\left[\,\cdot\thinspace,\cdot\,\right]$ stands for the horizontal
concatenation operator.

Let $L\in\mathbb{R}^{(m-r)\times m}$ denote a basis for the left null space
of $N$, which implies $LN=0$. We have $\mathrm{rank}(L)=m-r$.
We say that the system satisfies \emph{moiety conservation} if for
any initial concentration $c_{0}\in\mathbb{R}_{++}^{m}$, it holds
\[
L\,c=L\,\mathrm{exp}(x)=l_{0}
\]
along the trajectory of~\eqref{eq:dcdt2}, given an initial starting point $l_{0}\in\mathbb{R}_{++}^{m}$. It is possible to compute $L$ such that each row corresponds to a structurally identifiable conserved moiety in a biochemical network, cf.~\cite{Haraldsdttir2016}. The problem of finding the \emph{moiety conserved
steady state} of a biochemical reaction network is equivalent to solving
the nonlinear equation \eqref{eq:nonequa} with
\begin{equation}
h(x):=\left(\begin{array}{c}
\bar{f}(x)\\
L\,\mbox{exp}(x)-l_{0}
\end{array}\right).\label{eq:steadyStateEquation}
\end{equation}

By replacing $f$ by $\bar{f}$ we have improved the rank deficiency of $\nabla f$, and thus the one of $h$ in~\eqref{eq:steadyStateEquation}. Nonetheless, as we demonstrate in Figure~\ref{fig:rank}, $\nabla h$ is usually still far from being full rank at the solutions. 

Let us show that $h$ is real analytic. Let $A:=\left[\bar{N},-\bar{N}\right]$ and $B:=[F,\,R]^{T}$.
Then we can write
{\allowdisplaybreaks
\begin{align*}
\psi(x) & =\frac{1}{2}\|h(x)\|^{2}=\frac{1}{2}h(x)^{T}h(x)\\
 & =\frac{1}{2}\exp\left(k+Bx\right)^{T}A^{T}A\exp\left(k+Bx\right)\\
 &\quad+\frac{1}{2}\left(L\,\mathrm{exp}(x)-l_{0}\right)^{T}\left(L\,\mathrm{exp}(x)-l_{0}\right)\\
 & =\exp\left(k+Bx\right)^{T}Q\exp\left(k+Bx\right)+\frac{1}{2}\left(L\,\mathrm{exp}(x)-l_{0}\right)^{T}\left(L\,\mathrm{exp}(x)-l_{0}\right)\\
 & =\sum_{p,q=1}^{2n}Q_{pq}\exp\left(k_{p}+k_{q}+\sum_{i=1}^{m}(B_{pi}+B_{qi})x_{i}\right)\\
 &\quad+\frac{1}{2}\left(L\,\mathrm{exp}(x)-l_{0}\right)^{T}\left(L\,\mathrm{exp}(x)-l_{0}\right),
\end{align*}}%
where $Q=A^{T}A.$ Since~$B_{ij}$ are nonnegative
integers for all $i$ and~$j$, we conclude that the function~$\psi$
is real analytic (see Proposition~2.2.2 and Proposition~2.2.8 in~\cite{parks1992primer}).
It follows from Remark~\ref{rem:real_analytic} that $\psi$ satisfy the \L{}ojasiewicz gradient inequality (with some unknown exponent $\theta\in{[0,1[}$) and
the mapping~$h$ is H\"{o}lder metrically subregular around $(x^{*},0)$. Therefore, the assumptions of Theorem~\ref{thm:convergence_of_dist(x,Omega)_general} are satisfied as long as $\eta$ is \emph{sufficiently small}, and local sublinear convergence of~\ref{ALGORITHM-1:llm} is guaranteed.

\subsection{Computational experiments}

In this subsection, we compare \ref{ALGORITHM-1:llm} with various
Levenberg--Marquardt methods for solving the nonlinear system~\eqref{eq:nonequa} with $h$ defined by~\eqref{eq:steadyStateEquation} on 20 different biological models. These codes are available in the  COBRA Toolbox v3~\cite{COBRA}. In our implementation, all codes were written in MATLAB and runs
were performed on 
Intel Core i7-4770 CPU \@3.40GHz with 12GB RAM, under Windows 10 (64-bits).
The algorithms
were stopped whenever
\begin{equation*}
\begin{split}
\|h(x_{k})\|\leq 10^{-6}\label{eq:stopcrit}
\end{split}
\end{equation*}
is satisfied or the maximum number of iterations
(say 10,000) is reached. On the basis of our experiments with
the mapping~\eqref{eq:steadyStateEquation}, we set
\begin{equation}\label{eq:xi}
\xi_k:=\max\left\{0.95^{2k},10^{-9}\right\}\quad\text{and}\quad\omega_{k}:=0.95^{k}.
\end{equation}
The initial point is set to $x_{0}=0$ in all the experiments.

To illustrate the results, we use
the Dolan and Mor\'e performance profile~\cite{Dolan2002} with the performance measures
$N_{i}$ and $T$, where $N_i$ and $T$ denote the total number of iterations and the running time. In this procedure, the performance
of each algorithm is measured by the ratio of its computational outcome
versus the best numerical outcome of all algorithms. This performance
profile offers a tool to statistically compare the performance of algorithms. Let $\mathcal{S}$ be a set of all algorithms
and $\mathcal{P}$ be a set of test problems. For each problem $p$
and algorithm $s$, $t_{p,s}$ denotes the computational outcome with respect to the performance index, which is used in the definition of the performance ratio
\begin{equation}
r_{p,s}:=\frac{t_{p,s}}{\min\{t_{p,s}:s\in\mathcal{S}\}}.\label{eq:perf}
\end{equation}
If an algorithm $s$ fails to solve a problem $p$, the procedure
sets $r_{p,s}:=r_\text{failed}$, where $r_\text{failed}$ should be strictly
larger than any performance ratio~\eqref{eq:perf}. Let $n_p$ be the number of problems in the
experiment. For any factor
$\tau\in\mathbb{R}$, the overall performance of an algorithm $s$ is given by
\[
\rho_{s}(\tau):=\frac{1}{n_{p}}\textrm{size}\{p\in\mathcal{P}:r_{p,s}\leq\tau\}.
\]
Here, $\rho_{s}(\tau)$ is the probability that a performance ratio
$r_{p,s}$ of an algorithm $s\in\mathcal{S}$ is within a factor $\tau$
of the best possible ratio. The function $\rho_{s}(\tau)$ is a distribution
function for the performance ratio. In particular, $\rho_{s}(1)$
gives the probability that an algorithm $s$ wins over all other considered
algorithms, and $\lim_{\tau\rightarrow r_{\mathrm{failed}}}\rho_{s}(\tau)$
gives the probability that algorithm $s$ solves all considered
problems. Therefore, this performance profile can be considered as
a measure of efficiency among all considered algorithms. 

In our first experiment, we explore for which parameter
$\eta$ the best performance of \ref{ALGORITHM-1:llm} is obtained. To this end, we apply seven versions of \ref{ALGORITHM-1:llm} associated to each of the parameters $\eta\in\{0.6,0.7,0.8,0.9,0.99,\allowbreak 0.999,1\}$ to the
nonlinear system~\eqref{eq:steadyStateEquation} defined by 20 biological models.
The results of this comparison are summarised in Table~\ref{t.tuning_eta} and
Figure~\ref{fig:tuning_eta}, from where it can be observed
that \ref{ALGORITHM-1:llm} with $\eta=0.999$ outperforms the other values of the parameters. It is also apparent that smaller values of~$\eta$ are less efficient, although \ref{ALGORITHM-1:llm} successfully found a solution for every model and every value of $\eta$ that was tested. It is important to recall here that local convergence is only guaranteed by Theorem~\ref{thm:convergence_of_dist(x,Omega)_general} for \emph{sufficiently small} values of $\eta$, since the value of $\theta$ is unknown. Also, note that the local
convergence for the value $\eta=1$ is not covered by Theorem~\ref{thm:convergence_of_dist(x,Omega)_general} for our choice of the parameters, because it requires $\eta < \min\{1,2-2\theta\}$, since $\omega_{\min}=0$ in~\eqref{eq:xi}.
\begin{figure}[ht!]
\centering
\includegraphics[width=.65\textwidth]{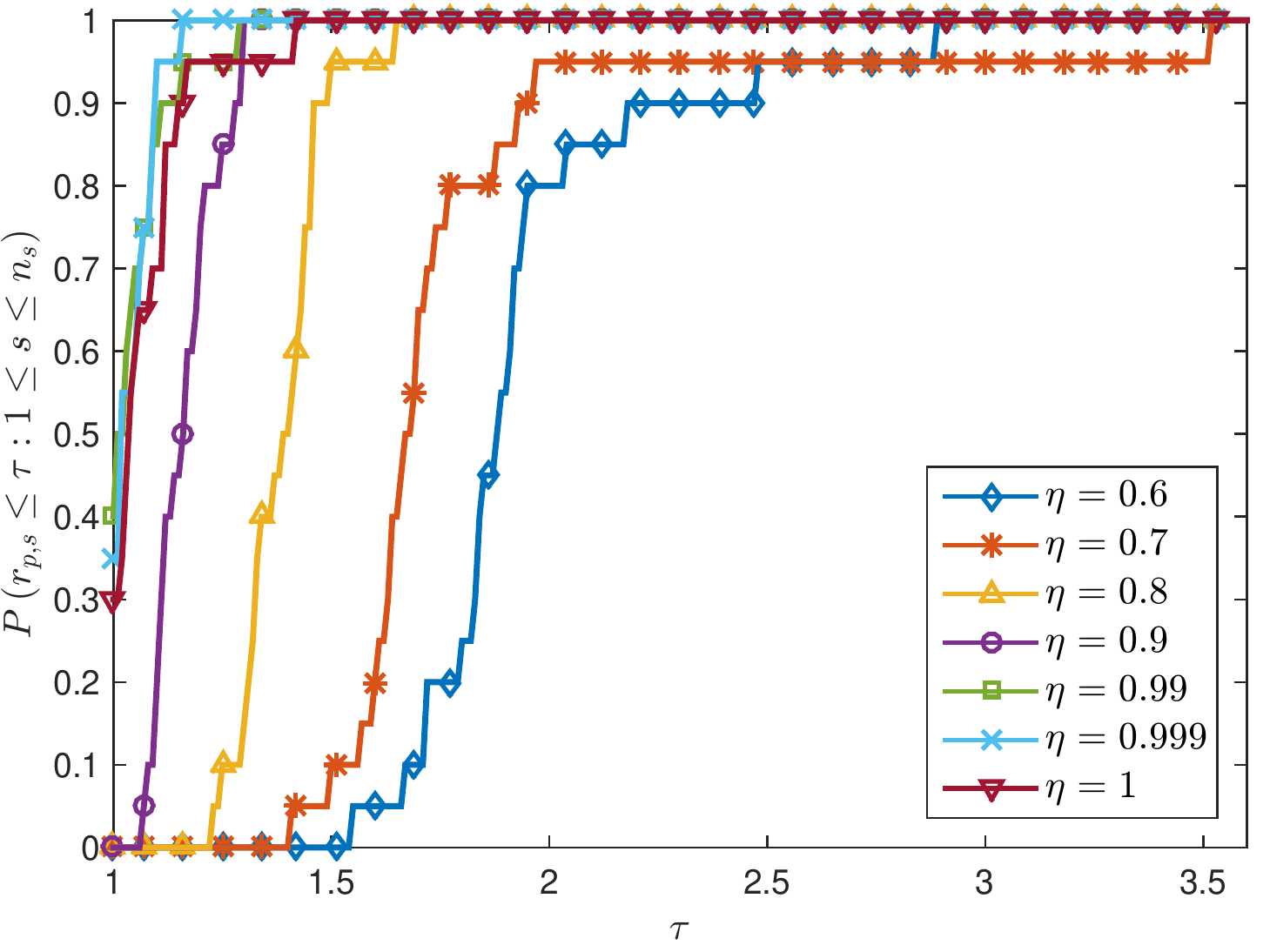}
\caption{Performance profile for the number of iterations of \ref{ALGORITHM-1:llm} with parameters~\eqref{eq:xi} and
$\eta\in\{0.6,0.7,0.8,0.9,\allowbreak 0.99,0.999,1\}$. The best performance is attained by $\eta=0.999$. \label{fig:tuning_eta}}
\end{figure}

We now set $\eta=0.999$ and compare \ref{ALGORITHM-1:llm} with parameters~\eqref{eq:xi} with the following Levenberg--Marquardt methods:
\renewcommand{\labelitemi}{$\bullet$}
\begin{itemize}
\item LM-YF: with $\mu_{k}=\|h(x_{k})\|^{2}$, given by Yamashita and Fukushima~\cite{alefeld_rate_2001};
\item LM-FY: with $\mu_{k}=\|h(x_{k})\|$, given by Fan and Yuan~\cite{fan_quadratic_2005};
\item LM-F: with $\mu_{k}=\|\nabla h(x_{k})h(x_{k})\|$, given by Fischer~\cite{fischer2002local}.
\end{itemize}
It is clear that all of these
three methods are special cases of \ref{ALGORITHM-1:llm} by selecting suitable parameters
$\xi_{k}$, $\omega_k$, and $\eta$. The results of our experiments are summarised
in Table~\ref{t.llm} and Figure~\ref{fig:per_pro_local}.
In Figures~\ref{fig:AnotherFig4a} and~\ref{fig:AnotherFig4b}, we see that \ref{ALGORITHM-1:llm} is clearly always the winner, both for the
number of iterations and the running time.
Moreover, LM-F outperforms both LM-YF and LM-FY. In fact, LM-FY was not able to solve any of the considered problems within the 10,000 iterations.

\begin{figure}[ht!]
\centering%
\subfigure[Number of iterations $N_i$]{\label{fig:AnotherFig4a}{\includegraphics[width=.48\textwidth]{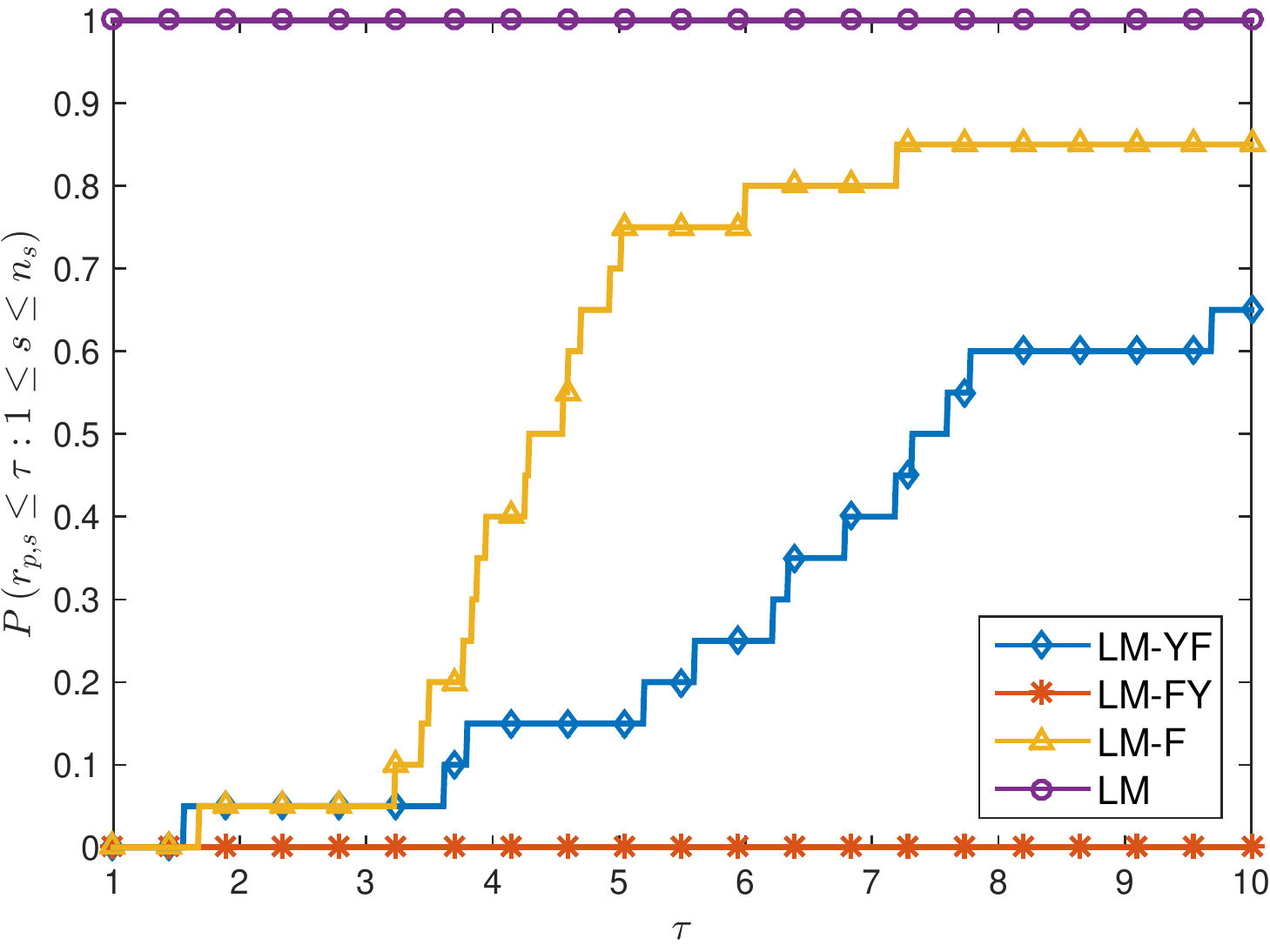}}}\quad
\subfigure[Running time $T$]{{\includegraphics[width=.48\textwidth]{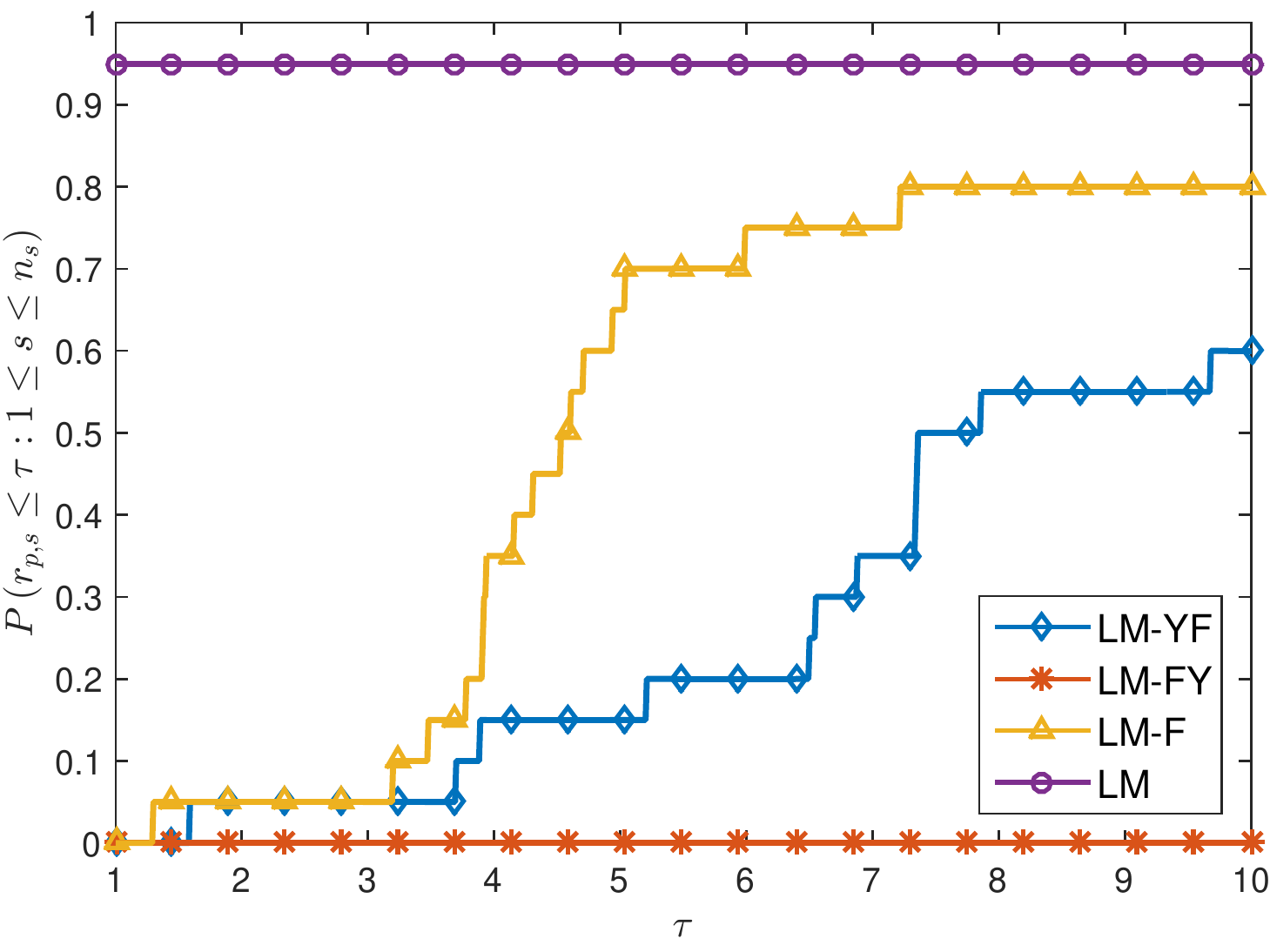}}\label{fig:AnotherFig4b}}%

\caption{Performance profiles for the number of iterations ($N_{i}$)
and the running time ($T$) of LM-YF, LM-FY, LM-F, and \ref{ALGORITHM-1:llm} with parameters~\eqref{eq:xi} and $\eta=0.999$ on a
set of 20 biological models for the mapping~\eqref{eq:steadyStateEquation}. \ref{ALGORITHM-1:llm} clearly outperforms the other methods.
\label{fig:per_pro_local}}
\end{figure}

In order to see the evolution of the merit function, we illustrate its value with respect to the number of iterations in Figure~\ref{fig:fan_vs_iter} for the mapping~\eqref{eq:steadyStateEquation} with the biological models  iAF692 and iNJ661.
We limit the maximum number of iterations to 1,000. Clearly, \ref{ALGORITHM-1:llm} attains the best results, followed by LM-F. Both methods seem to be more suited to biological problems than LM-YF and LM-FY. We also show in Figure~\ref{fig:fan_vs_iter} the evolution of the value of the step size $\|d_k\|$. Both \ref{ALGORITHM-1:llm} and LM-F show a rippling behaviour, while the value of $\|d_k\|$ is nearly constant along the 1,000 iterations for LM-YF and LM-FY. Probably, this rippling behaviour is letting the first two methods escape from a flat valley of the merit function, while the two last methods get trapped there. Observe also that, by Lemma~\ref{Svaiter_lemma}, one has that $\|d_k\|\leq\frac{1}{2}$ for LM-YF and $\|d_k\|\leq\frac{1}{2}\|h(x_k)\|^\frac{1}{2}$ for LM-FY, while this upper bound can be larger for both \ref{ALGORITHM-1:llm} and LM-F.

\begin{figure}[ht!]
\centering%
\subfigure[Merit function of iAF692]{\label{fig:AnotherFig}{\includegraphics[width=.49\textwidth]{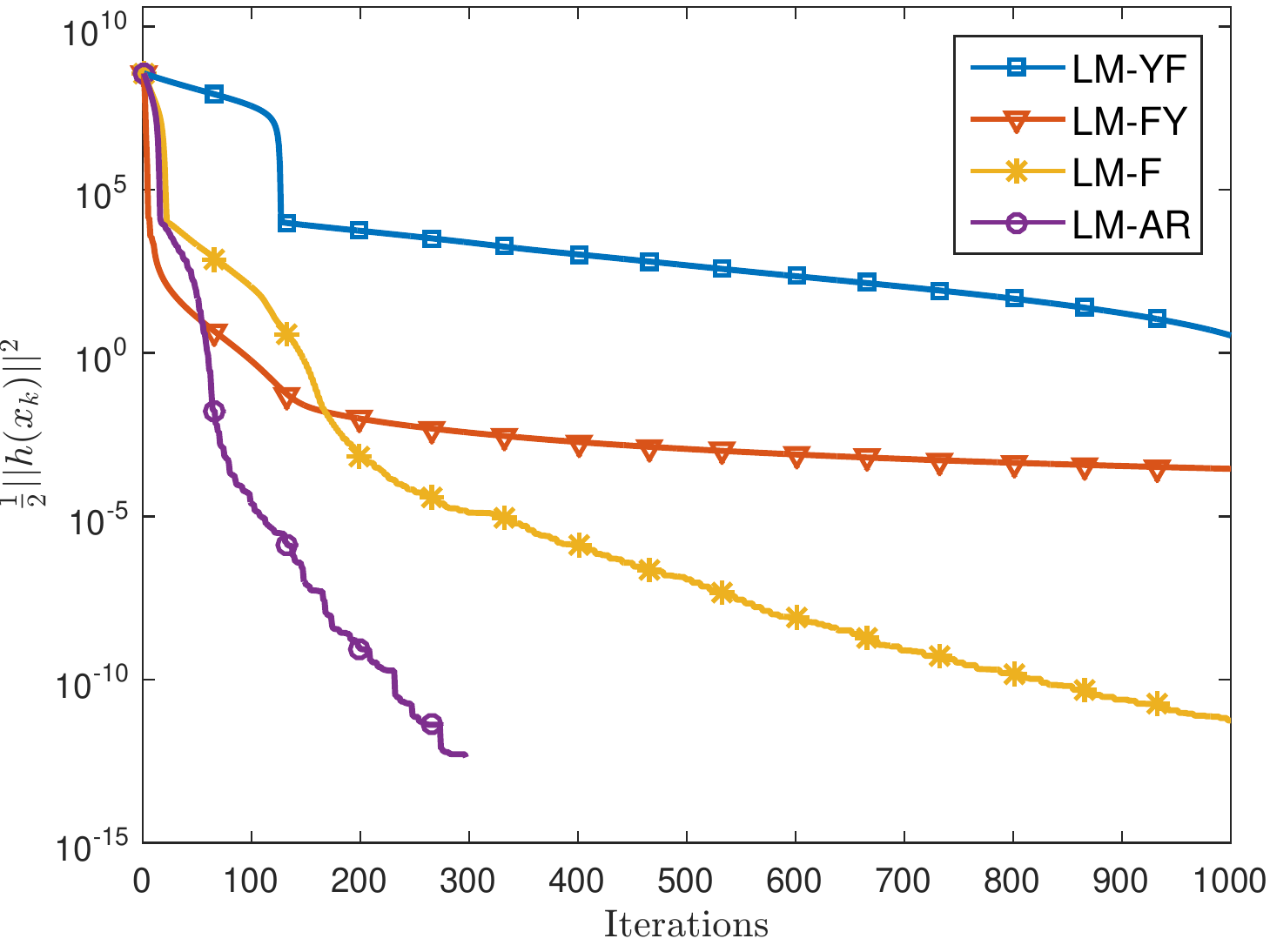}}}\hspace{4pt}
\subfigure[Merit function of iNJ661]{{\includegraphics[width=.49\textwidth]{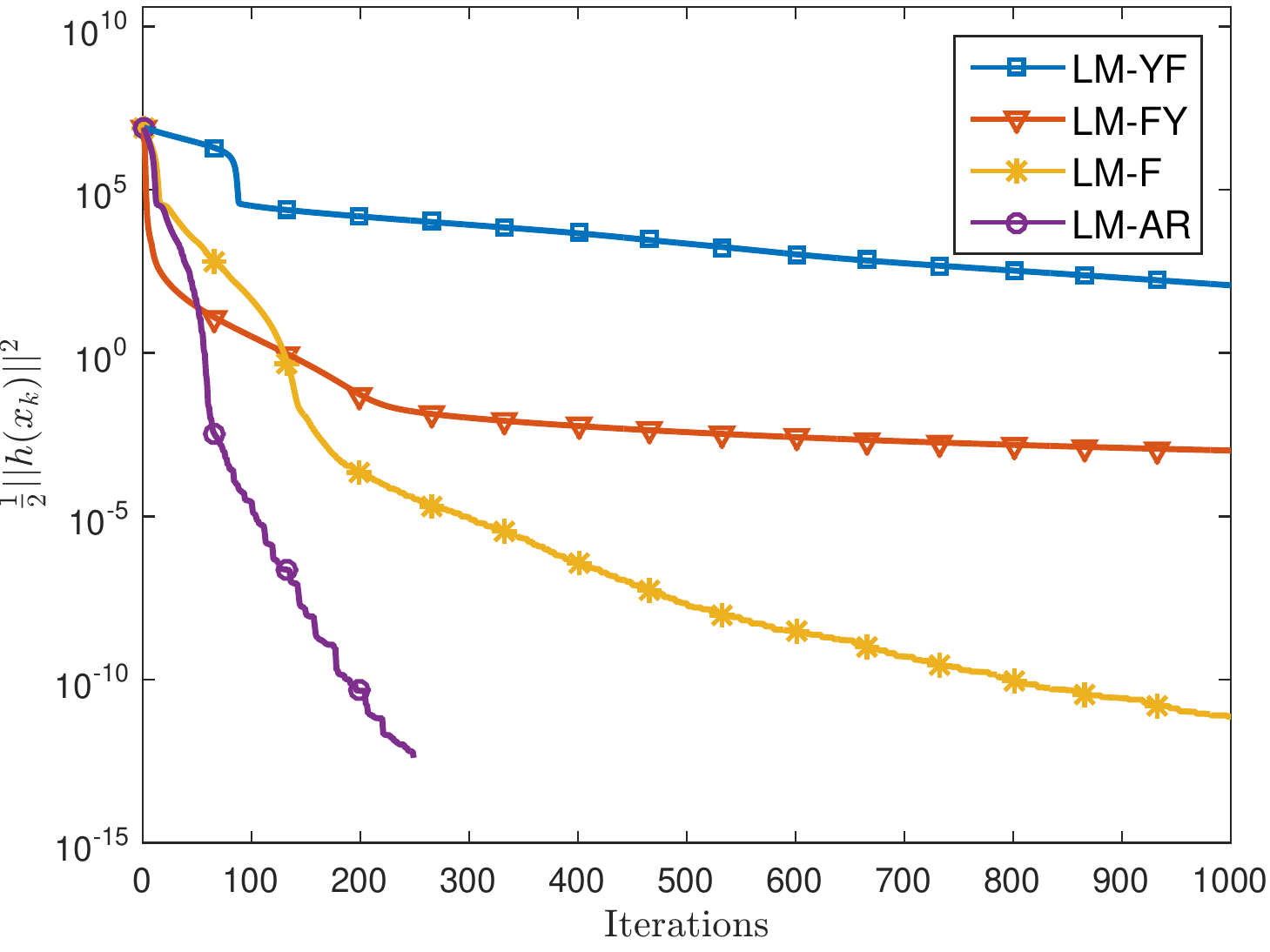}}\label{fig:AnotherFig2}}
\subfigure[Step sizes for iAF692]{\label{fig:AnotherFig3}{\includegraphics[width=.49\textwidth]{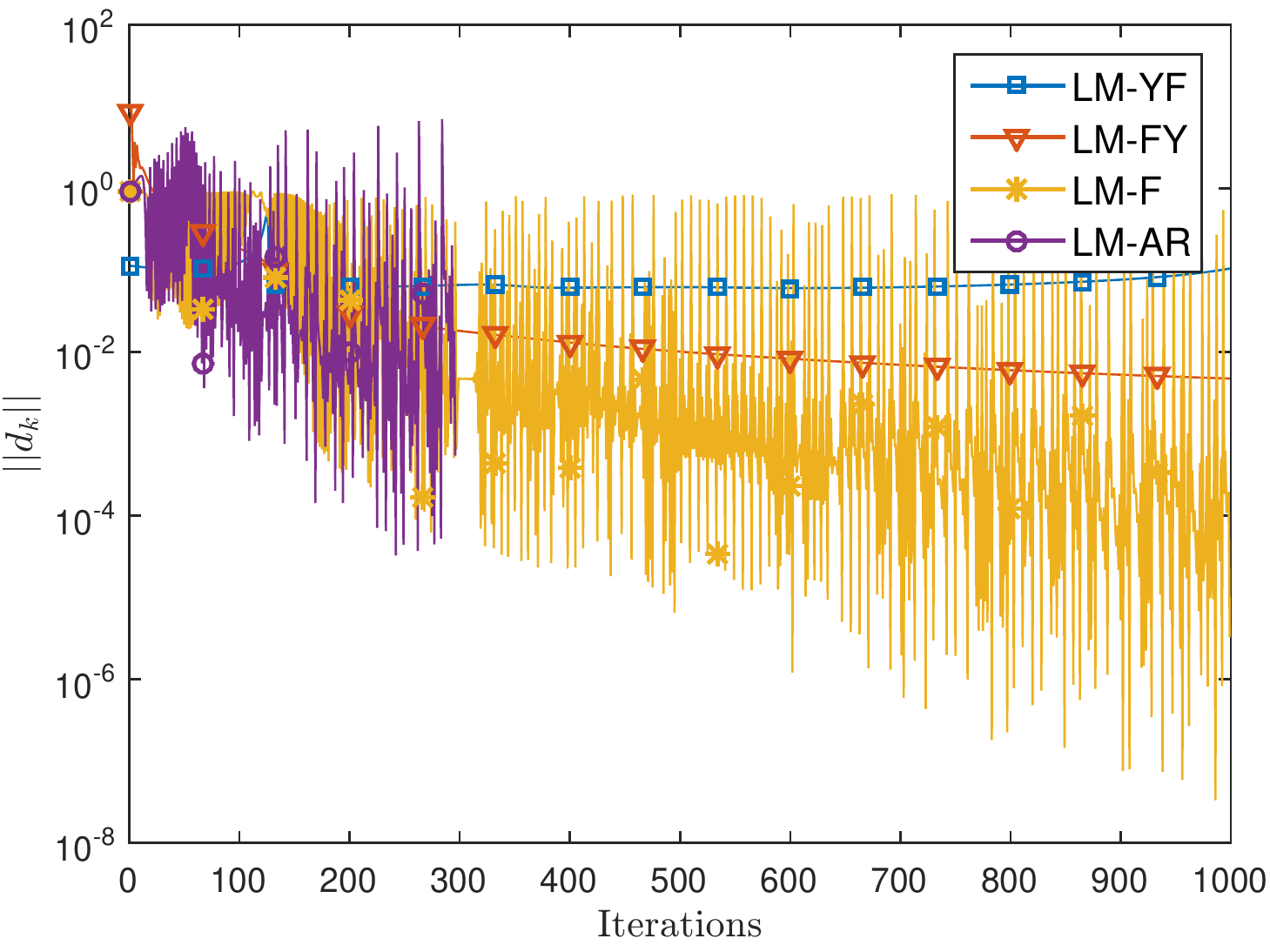}}}\hspace{4pt}
\subfigure[Step sizes for iNJ661]{{\includegraphics[width=.49\textwidth]{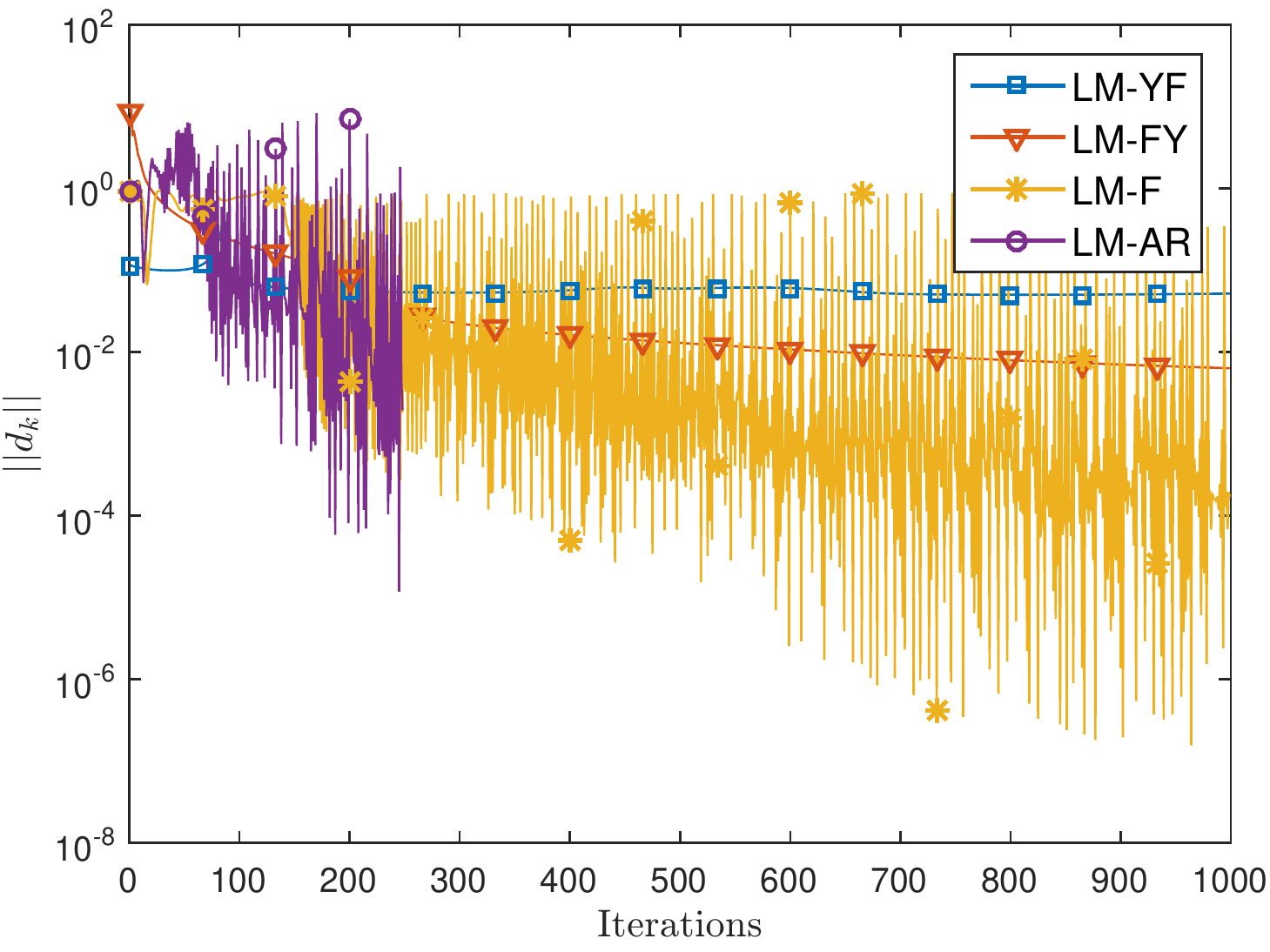}}\label{fig:AnotherFig4}}
\vspace{-2mm}
\caption{Value of the merit function and step size with respect to the number of iterations for the methods LM-YF, LM-FY, LM-F, and \ref{ALGORITHM-1:llm} with parameters~\eqref{eq:xi} and $\eta=0.999$, when applied to the mapping~\eqref{eq:steadyStateEquation} defined by the biological
models iAF692 and iNJ661. It clearly shows that~\ref{ALGORITHM-1:llm} outperforms the other methods. \label{fig:fan_vs_iter}}
\end{figure}

In our last experiment, we find $10$ solutions of the nonlinear system~\eqref{eq:nonequa}  with~\ref{ALGORITHM-1:llm} using $10$~random starting points $x_0\in{\left]-\frac{1}{2},\frac{1}{2}\right[^m}$ for each of the $20$ biological models and compute the rank of $\nabla h$ at each of these solutions. The results are shown in Figure~\ref{fig:rank}, where we plot the rank deficiency of $\nabla h$ at each of the solutions. For all the models, except for the Ecoli\_core, we observe that $\nabla h$ at the solutions found is far from being full rank. For the Ecoli\_core, although $\nabla h$ had full rank at every solution found, the smallest eigenvalue at these solutions had a value around $10^{-9}$, making also this problem ill-conditioned. This explains the difficulties that most of the algorithms had for solving the nonlinear system~\eqref{eq:nonequa} with $h$ defined by~\eqref{eq:steadyStateEquation}. Therefore, since we are dealing with a difficult problem, it is more meritorious the successfulness of \ref{ALGORITHM-1:llm} with parameters~\eqref{eq:xi} for finding a solution of each of the $20$ models in less than 400 iterations (in less than one minute), as shown in Table~\ref{t.llm}.
\begin{figure}[ht!]
\centering%
\includegraphics[width=.7\textwidth]{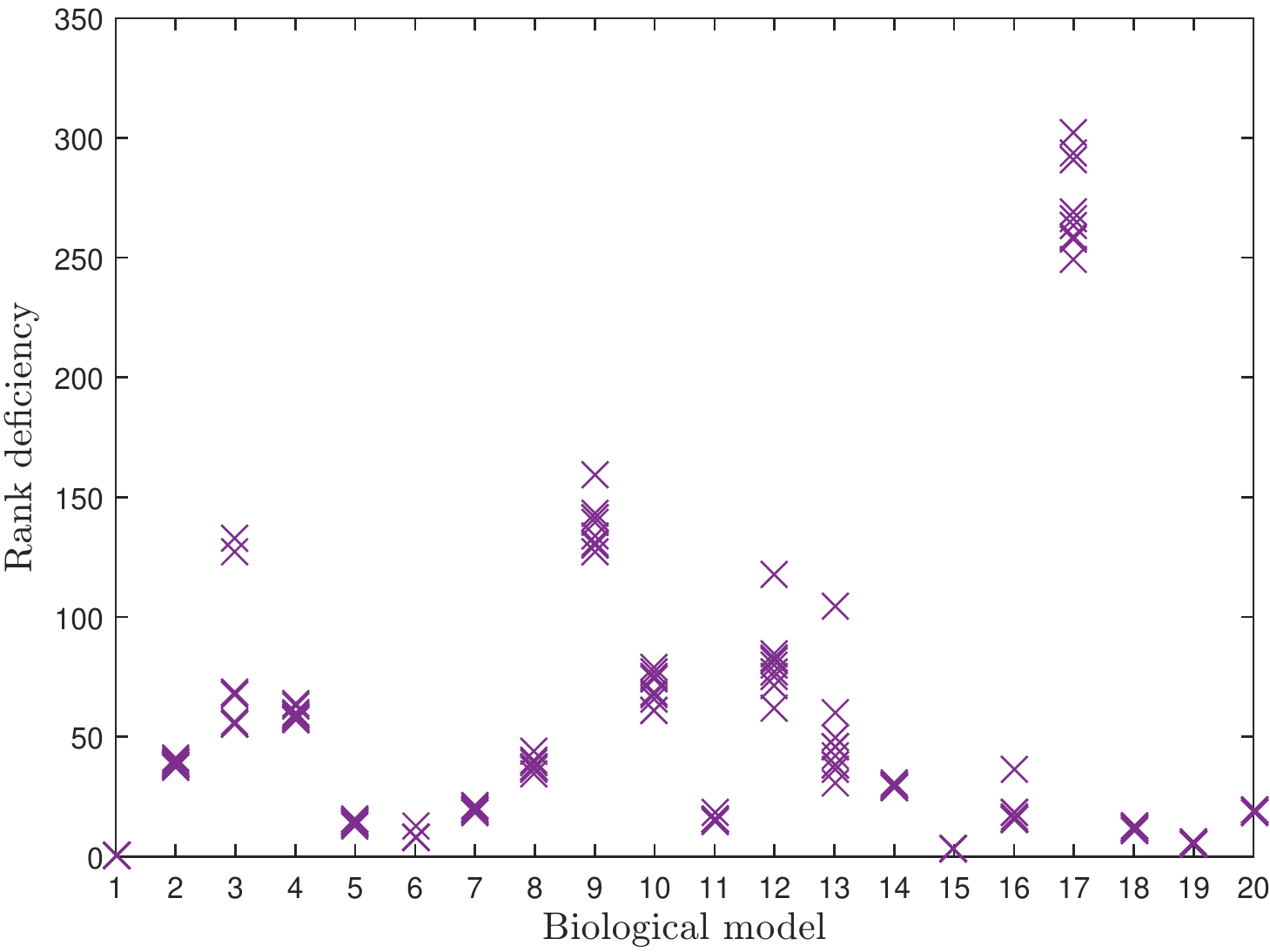}
\caption{Plot of the difference between $m$ and the rank of $\nabla h$ at $10$ solutions found with \ref{ALGORITHM-1:llm} for each of the $20$ biological models considered.  The models are represented in the $x$-axis, using the same order than in Tables~\ref{t.tuning_eta} and~\ref{t.llm}.\label{fig:rank}}
\end{figure}

\section{Conclusion and further research}\label{sec:conclusion}
We have presented an adaptive Levenberg--Marquardt method for solving
systems of nonlinear equations with possible non-isolated solutions. We have analysed its local convergence under H\"{o}lder metric subregularity of the underlying function and H\"older continuity of its gradient. We have further analysed the local convergence under the additional assumption that the \L{}ojasiewicz gradient inequality holds. These properties hold in many applied problems, as they are satisfied by any real analytic function. One of these applications is computing a solution to a system of nonlinear equations arising in biochemical reaction networks, a problem which is usually ill-conditioned. We showed that such systems satisfy both the H\"{o}lder metric subregularity and the \L{}ojasiewicz gradient inequality assumptions. In our numerical experiments, we clearly obtained a superior performance of our regularisation parameter, compared to existing Levenberg--Marquardt methods, for 20 different biological networks.

Several extensions to the present study are possible, the most important of which would be
to develop a globally convergent version of the proposed Levenberg--Marquardt method. One approach, which is currently being investigated, would be to combine the scheme with an Armijo-type line search and a trust-region technique. This will be reported in a separate article~\cite{AFV19}. It would also be interesting to analyse a regularisation parameter where the value of $\eta$ is updated at each iteration. The analysis of the convergence with such a parameter would be much more involved, so we leave this for future work.

\section*{Acknowledgements}
We would like to thank Mikhail Solodov for suggesting the use of Levenberg--Marquardt methods for solving the system of nonlinear equations arising in biochemical reaction networks. Thanks also go to Michael Saunders for his useful comments on the first version of this manuscript. We are grateful to two anonymous reviewers for their constructive comments, which helped us improving the paper.


\section*{Appendix} See Tables~\ref{t.tuning_eta} and~\ref{t.llm} for the summary results of the comparisons.

\begin{landscape}
\begin{table}[htbp]
\caption{Summary of the results of tuning the parameter $\eta$ for \ref{ALGORITHM-1:llm} with parameters~\eqref{eq:xi} and
$\eta\in\{0.6, 0.7, 0.8, 0.9, 0.99, 0.999, 1\}$ to solve~\eqref{eq:steadyStateEquation}
in 20~biological models. For each model, the lowest number of iterations ($N_i$) and the lowest running time ($T$) are displayed in bold.}
\label{t.tuning_eta}
\begin{center}\scriptsize
\renewcommand{\arraystretch}{1.3}
\begin{tabular}{|l|rrr|rrrrrrrrrrrrrr|}
\hline
\multirow{2}{*}{Model} &\multirow{2}{*}{$m$} &\multirow{2}{*}{$n$} &\multirow{2}{*}{$r$} &
\multicolumn{2}{c}{$\eta=0.6$} &
\multicolumn{2}{c}{$\eta=0.7$} &
\multicolumn{2}{c}{$\eta=0.8$} &
\multicolumn{2}{c}{$\eta=0.9$} &
\multicolumn{2}{c}{$\eta=0.99$} &
\multicolumn{2}{c}{$\eta=0.999$} &
\multicolumn{2}{c|}{$\eta=1$} \\
 & &    & & $N_i$ & $T$ & $N_i$ & $T$ & $N_i$ & $T$ & $N_i$ & $T$ & $N_i$ & $T$ & $N_i$ & $T$ & $N_i$ & $T$\\
\hline
1. Ecoli\_core &72&73&61&210&0.10&191&0.08&179&0.07&159&0.06&\bf 136&0.06&138&\bf 0.04&139&0.09\\
2. iAF692 &462&493&430&492&6.89&421&5.70&328&4.54&291&3.98&274&3.72&256&3.44&\bf 253&\bf 3.42\\
3. iAF1260 &1520&1931&1456&473&98.75&410&83.41&357&74.05&334&68.65&\bf 257&\bf 53.69&271&57.05&268&54.88\\
4. iBsu1103 &993&1167&956&421&34.34&356&28.46&313&25.14&254&20.34&232&18.38&\bf 218&\bf 17.13&226&18.04\\
5. iCB925 &415&558&386&467&6.41&477&6.47&332&4.55&296&4.30&318&4.39&\bf 248&\bf 3.33&350&4.71\\
6. iIT341 &424&428&392&388&4.33&333&3.48&253&2.68&226&2.29&\bf 207&\bf 2.13&226&2.34&212&2.19\\
7. iJN678 &641&669&589&362&9.63&356&8.99&307&7.76&258&6.72&220&5.76&231&5.89&\bf 218&\bf 5.60\\
8. iJN746 &727&795&700&470&17.05&376&13.29&301&10.66&255&8.93&256&9.26&\bf 231&\bf 8.17&251&9.10\\
9. iJO1366 &1654&2102&1582&417&107.37&372&95.21&314&79.85&273&70.62&225&57.36&\bf 219&\bf 55.65&244&62.37\\
10. iJR904 &597&757&564&441&11.99&420&11.15&346&8.80&279&7.26&\bf 245&\bf 6.34&249&6.39&253&6.98\\
11. iMB745 &525&598&490&363&7.13&346&6.48&298&5.55&211&3.85&204&3.74&218&4.02&\bf 199&\bf 3.66\\
12. iNJ661 &651&764&604&549&16.21&436&12.68&366&11.01&283&8.30&\bf 222&\bf 6.45&257&7.53&254&7.57\\
13. iRsp1095 &966&1042&921&655&54.48&1057&92.98&374&28.20&333&25.03&\bf 301&\bf 21.89&301&22.63&336&25.52\\
14. iSB619 &462&508&435&373&5.13&344&4.60&295&3.95&243&3.24&\bf 203&\bf 2.73&221&2.97&215&2.85\\
15. iTH366 &583&606&529&349&7.27&338&6.91&279&5.67&219&4.46&212&4.33&\bf 204&\bf 4.16&207&4.21\\
16. iTZ479\_v2 &435&476&415&375&4.54&344&4.00&297&3.45&228&2.68&\bf 204&\bf 2.39&214&2.48&227&2.63\\
17. iYL1228 &1350&1695&1280&846&136.52&550&87.18&387&60.96&326&51.24&294&46.28&318&49.99&\bf 293&\bf 46.18\\
18. L\_lactis\_MG1363 &483&491&429&463&6.54&427&5.89&351&4.80&314&4.36&279&3.85&\bf 242&\bf 3.35&282&3.87\\
19. Sc\_thermophilis\_rBioNet &348&365&320&413&2.85&377&2.57&344&2.31&291&1.97&252&1.66&256&1.69&\bf 241&\bf 1.63\\
20. T\_Maritima &434&470&414&407&4.83&318&3.72&281&3.26&233&2.70&231&2.65&215&2.49&\bf 212&\bf 2.44\\
\hline
Average &  &  &
&447&27.12&412&24.16&315&17.36&265&15.05&239&12.85&237&13.04&244&13.4\\
\hline
\end{tabular}
\end{center}
\end{table}
\end{landscape}

\begin{landscape}
\begin{table}[htbp]
\caption{Summary of the results of LM-YF, LM-FY, LM-F, and \ref{ALGORITHM-1:llm} with parameters~\eqref{eq:xi} and $\eta=0.999$ for solving~\eqref{eq:steadyStateEquation} in 20 biological models. For each model, the lowest number of iterations ($N_i$) and the lowest running time ($T$) are displayed in bold. On the bottom, we show the average among the successfully solved instances.}
\label{t.llm}
\begin{center}\normalsize
\renewcommand{\arraystretch}{1.3}
\begin{tabular}{|l|rrr|rrrrrrrr|}
\hline
\multirow{2}{*}{Model} & \multirow{2}{*}{$m$} & \multirow{2}{*}{$n$} & \multirow{2}{*}{$r$} &
\multicolumn{2}{c}{LM-YF} &
\multicolumn{2}{c}{LM-FY} &
\multicolumn{2}{c}{LM-F} &
\multicolumn{2}{c|}{\ref{ALGORITHM-1:llm}} \\
& & & & $N_i$ & $T$ & $N_i$ & $T$ & $N_i$ & $T$ & $N_i$ & $T$\\
\hline
1. Ecoli\_core &72&73&61&238&0.10&10000&4.06&257&0.08&\bf 153&\bf 0.06\\
2. iAF692 &462&493&430&1685&23.43&10000&135.63&1358&18.15&\bf 271&\bf 3.61\\
3. iAF1260 &1520&1931&1456&8233&1726.92&10000&2066.38&2036&413.04&\bf 283&\bf 57.27\\
4. iBsu1103 &993&1167&956&2396&187.24&10000&780.72&761&59.39&\bf 193&\bf 15.09\\
5. iCB925 &415&558&386&1005&14.15&10000&131.01&10000&131.28&\bf 278&\bf 3.82\\
6. iIT341 &424&428&392&1407&14.77&10000&103.09&944&9.71&\bf 222&\bf 2.26\\
7. iJN678 &641&669&589&2218&56.33&10000&253.58&1043&26.34&\bf 229&\bf 5.82\\
8. iJN746 &727&795&700&3107&108.72&10000&349.43&1069&37.18&\bf 217&\bf 7.53\\
9. iJO1366 &1654&2102&1582&7716&1946.48&10000&2524.01&1066&268.38&\bf 232&\bf 58.33\\
10. iJR904 &597&757&564&2789&72.26&10000&258.50&1231&31.74&\bf 262&\bf 6.75\\
11. iMB745 &525&598&490&790&14.60&10000&181.40&1247&22.50&\bf 208&\bf 3.76\\
12. iNJ661 &651&764&604&2635&76.62&10000&290.56&1357&39.45&\bf 360&\bf 10.44\\
13. iRsp1095 &966&1042&921&3832&266.87&10000&694.21&10000&696.93&\bf 235&\bf 16.38\\
14. iSB619 &462&508&435&1581&21.42&10000&133.94&814&10.82&\bf 233&\bf 3.12\\
15. iTH366 &583&606&529&1641&33.78&10000&205.34&817&16.79&\bf 211&\bf 4.30\\
16. iTZ479\_v2 &435&476&415&1148&13.57&10000&117.46&713&8.34&\bf 221&\bf 2.61\\
17. iYL1228 &1350&1695&1280&6070&956.92&10000&1565.75&10000&1567.47&\bf 272&\bf 42.94\\
18. L\_lactis\_MG1363 &483&491&429&2180&30.17&10000&137.99&1231&17.06&\bf 287&\bf 4.10\\
19. Sc\_thermophilis\_rBioNet &348&365&320&1753&11.85&10000&68.54&935&6.33&\bf 244&\bf 1.61\\
20. T\_Maritima &434&470&414&1169&14.33&10000&118.65&717&8.31&\bf 209&\bf 2.42\\
\hline
Average of successful&  &  &
& 2680&279.53&\multicolumn{1}{c}{---}&\multicolumn{1}{c}{---}&1035&58.45&\bf 241&\bf 12.61\\
\hline
\end{tabular}
\end{center}
\end{table}
\end{landscape}




\end{document}